\RequirePackage{amsmath}

\documentclass[runningheads]{llncs}
\usepackage[title]{appendix}
\usepackage[normalem]{ulem}

\usepackage[T1]{fontenc}
\usepackage{amsfonts}
\usepackage{amsmath}
\usepackage{hyperref}
\hypersetup{
    colorlinks=true,
    linkcolor=blue,     
    urlcolor=blue,
    pdfpagemode=FullScreen,
    }

\usepackage{graphicx}
\usepackage{xcolor}

\let\epsilon\varepsilon

\DeclareMathOperator{\poly}{poly}

\usepackage{tikz}
\usetikzlibrary{decorations.pathreplacing}
\usetikzlibrary{plotmarks}
\usetikzlibrary{positioning,automata,arrows}
\usetikzlibrary{shapes.geometric}
\usetikzlibrary{decorations.markings}
\usetikzlibrary{positioning, shapes, arrows}

\PassOptionsToPackage{usenames,dvipsnames,svgnames}{xcolor}
\definecolor{orange}{RGB}{235,90,0}
\definecolor{darkorange}{RGB}{175,30,0}
\definecolor{turkis}{RGB}{131,182,182}
\definecolor{darkturkis}{RGB}{31,82,82}
\definecolor{green}{RGB}{102,180,0}
\definecolor{darkgreen}{RGB}{51,90,0}
\definecolor{myblue}{RGB}{0,0,213}
\definecolor{mydarkblue}{RGB}{0,0,100}
\definecolor{mybrightblue}{HTML}{74B0E4}
\definecolor{mybrighterblue}{HTML}{B3EAFA}
\definecolor{lila}{RGB}{102,0,102}
\definecolor{darkred}{RGB}{139,0,0}
\definecolor{darkyellow}{RGB}{188,135,2}
\definecolor{brightgray}{RGB}{200,200,200}
\definecolor{darkgray}{RGB}{50,50,50}
\definecolor{amaranth}{rgb}{0.9, 0.17, 0.31}
\definecolor{alizarin}{rgb}{0.82, 0.1, 0.26}
\definecolor{amber}{rgb}{1.0, 0.75, 0.0}
\definecolor{green(ryb)}{rgb}{0.4, 0.69, 0.2}
\definecolor{hanblue}{rgb}{0.27, 0.42, 0.81}
\definecolor{grannysmithapple}{rgb}{0.66, 0.89, 0.63}

\begin{document}
\title{Optimal Wheeler Language Recognition
\thanks{The conference version of this paper can be found in \cite{spire23paper}.\\
Ruben Becker, Davide Cenzato, Sung-Hwan Kim, Bojana Kodric, and Nicola Prezza are funded by the European Union (ERC, REGINDEX, 101039208). Views and opinions expressed are however those of the author(s) only and do not necessarily reflect those of the European Union or the European Research Council. Neither the European Union nor the granting authority can be held responsible for them. Alberto Policriti is supported by project National Biodiversity Future Center-NBFC (CN\_00000033, CUP G23C22001110007) under the National Recovery and Resilience Plan  of Italian Ministry of University and
Research funded by European Union--NextGenerationEU.}}
%
%
\author{
Ruben {Becker}\inst{1}\orcidID{0000-0002-3495-3753}
\and
Davide {Cenzato}\inst{1}\orcidID{0000-0002-0098-3620}
\and 
Sung-Hwan {Kim}\inst{1}\orcidID{0000-0002-1117-5020}
\and
Bojana {Kodric}\inst{1}\orcidID{0000-0001-7242-0096}
\and
Alberto {Policriti}\inst{2}\orcidID{0000-0001-8502-5896}
\and
Nicola {Prezza}\inst{1}\orcidID{0000-0003-3553-4953}
}
\authorrunning{R. Becker et al.}
%
\institute{
	Ca' Foscari University of Venice, Italy\\ \email{\{rubensimon.becker,davide.cenzato,sunghwan.kim,\\bojana.kodric,nicola.prezza\}@unive.it \vspace{5pt}}
\and
	University of Udine, Italy\\ \email{alberto.policriti@uniud.it}
}
\maketitle              
\begin{abstract}
A Wheeler automaton is a finite state automaton whose states admit a total \emph{Wheeler order}, reflecting the co-lexicographic order of the strings labeling source-to-node paths. A \emph{Wheeler language} is a regular language admitting an accepting Wheeler automaton. 
Wheeler languages admit efficient and elegant solutions to hard problems such as automata compression and regular expression matching, therefore deciding whether a regular language is Wheeler is relevant in applications requiring efficient solutions to those problems.
In this paper, we show that it is possible to
decide whether a DFA with $n$ states and $m$ transitions recognizes a Wheeler language in $O(mn)$ time. This is a significant improvement over the running time $O(n^{13} + m\log n)$ of the previous polynomial-time algorithm (Alanko et al., Information and Computation 2021). 
A proof-of-concept implementation of this algorithm is available in a public repository.
We complement this upper bound with a conditional matching lower bound stating that, unless the strong exponential time hypothesis (SETH) fails, the problem cannot be solved in strongly subquadratic time. 
The same problem is known to be
PSPACE-complete when the input is an NFA (D'Agostino et al., Theoretical Computer Science 2023). Together with that result, our paper essentially closes the algorithmic problem of Wheeler language recognition. 

\end{abstract}

\section{Introduction}

Wheeler automata were introduced by Gagie et al.\ in \cite{gagie:tcs17:wheeler} as a natural generalization of prefix-sorting techniques (standing at the core of the most successful string processing algorithms) to labeled graphs. Informally speaking, an automaton on alphabet $\Sigma$ is Wheeler if the co-lexicographic order of the strings labeling source-to-states paths can be ``lifted'' to a \emph{total} order of the states (a formal definition is given in Definition \ref{def:WDFA}).
As shown by the authors of \cite{gagie:tcs17:wheeler}, Wheeler automata can be encoded in just $O(\log|\Sigma|)$ bits per edge and they support near-optimal time pattern matching queries (i.e.\ finding all nodes reached by a path labeled with a given query string). These properties make them a powerful tool in fields such as bioinformatics, where one popular way to cope with the rapidly-increasing number of available fully-sequenced genomes, is to encode them in a pangenome graph: aligning short DNA sequences allows one to discover whether the sequences at hand contain variants recorded (as sub-paths) in the graph \cite{pangenomeGraphs}. 

Wheeler languages --- that is, regular languages recognized by Wheeler automata --- were later studied by  Alanko et al.\ in \cite{alanko:iac21:wheeler}. In that paper, the authors showed that Wheeler DFAs (WDFAs) and Wheeler NFAs (WNFAs) have the same expressive power. As a matter of fact, the class of Wheeler languages proved to possess several other remarkable properties, in addition to represent the class of regular languages for which efficient indexing data structures exist. Such properties motivated them to study the following decisional problem (as well as the corresponding variant on NFAs / regular expressions): 

\begin{definition}[\textsc{WheelerLanguageDFA}]
Given a DFA $\mathcal A$, decide if  the regular language $\mathcal L(\mathcal A)$ recognized by $\mathcal A$ is Wheeler. 
\end{definition}

Alanko et al.~\cite{alanko:iac21:wheeler}
provided the following characterization: 
a language $\mathcal L$ is Wheeler if and only if, for any co-lexicographically monotone sequence of strings $\alpha_1 \prec \alpha_2 \prec \dots$ (or with reversed signs $\succ$) belonging to the prefix-closure of $\mathcal L$,
on the minimum DFA for $\mathcal L$ there 
exists some $N \in \mathbb N$ and state $u$ such that 
by reading 
$\alpha_i$ from the source state we end up in state $u$ for all $i\geq N$. This characterization allowed them to 
devise a polynomial-time algorithm solving \textsc{WheelerLanguageDFA}.
This result is not trivial for two main reasons: (1)  
the smallest WDFA for a Wheeler language  $\mathcal L$ could be \emph{exponentially larger} than the smallest DFA  for $\mathcal L$ \cite{alanko:iac21:wheeler},
and 
(2) the corresponding \textsc{WheelerLanguageNFA} problem
(i.e., the input $\mathcal A$ is an NFA)
is PSPACE-complete \cite{DAGOSTINO2023113709}.

\paragraph{\textbf{Our Contributions}}

Despite being polynomial, the algorithm of Alanko et al. has a prohibitive time complexity: $O(n^{13}+m\log n)$, where $m$ and $n$ are the number of transitions and states of the input DFA\footnote{While the authors only claim $m^{O(1)}$ time, a finer analysis yields this bound.}. 
In this paper, we present a much simpler parameterized (worst-case quadratic) algorithm solving \textsc{WheelerLanguageDFA}. The complexity of our algorithm depends on a parameter $p$ --- the \emph{co-lex width} of the minimum DFA $\mathcal A_{min}$ for the language \cite{cotumaccio:soda21:psortable} (Definition \ref{def:colex width}), which is never larger than $n$ and which measures the ``distance'' of $\mathcal A_{min}$ from being Wheeler; e.g., if $\mathcal A_{min}$ is itself Wheeler, then $p=1$. We prove:

\begin{theorem}\label{thm:check wheeler}
\label{theorem: main result}
    \textsc{WheelerLanguageDFA} can be solved in $O(mp + m\log n) \subseteq O(mn)$ time on any DFA $\mathcal A$ with $n$ states and $m$ edges, where $p\leq n$ is the co-lex width of the minimum automaton $\mathcal A_{min}$ equivalent to $\mathcal A$.
\end{theorem}

The intuition behind Theorem \ref{theorem: main result} is the following. Starting from the characterization of Wheeler languages of Alanko et al. \cite{alanko:iac21:wheeler} based on monotone sequences, we show that $\mathcal L(\mathcal A)$ is not Wheeler if and only if the \emph{square automaton} $\mathcal A_{min}^2 = \mathcal A_{min} \times \mathcal A_{min}$ contains a cycle $(u_1,v_1) \rightarrow (u_2,v_2) \rightarrow \dots \rightarrow (u_k,v_k) \rightarrow (u_1,v_1)$ such that, for all $i=1, \dots, k$, the following two properties hold: (i) $u_i\neq v_i$ and (ii) 
the co-lexicographic ranges of strings reaching $u_i$ and $v_i$ intersect.
As a result, after computing $\mathcal A_{min}$ ($O(m\log n)$ time by Hopcroft's algorithm) and directly building this ``pruned'' version of $\mathcal A_{min}^2$ in $O(mp + m\log n)$ time using recent techniques described in \cite{kop:cpm23:working,becker2023sorting}, testing its acyclicity yields the answer. 
A proof-of-concept implementation of the algorithm behind Theorem \ref{theorem: main result} is available at \url{https://github.com/regindex/Wheeler-language-recognizer}. 

\medskip

We complement the above upper bound with a matching conditional lower bound. Our lower bound is obtained via a reduction from the following problem:

\begin{definition}[Orthogonal Vectors problem (OV)]\label{def:OV}
Given two sets $A$ and $B$, each containing $N$ vectors from $\{0,1\}^d$, decide whether there exist $a\in A$ and $b\in B$ such that $a^Tb = 0$. 
\end{definition}

By a classic reduction \cite{Williams05}, for $d\in\omega(\log N)$ OV cannot be solved in time $O(N^{2-\eta} \poly(d))$ for any constant $\eta>0$
unless the strong exponential time hypothesis \cite{impagliazzo2001complexity} (SETH) fails. We prove \footnote{Our lower bound states that there is no algorithm solving \emph{all} instances in $O(m^{2-\eta})$ time. On sparse DFAs ($m \in \Theta(n)$) our algorithm runs in $O(mn) = O(m^2)$ time.} :

\begin{theorem}\label{thm:lower bound}
If \textsc{WheelerLanguageDFA} can be solved in time $O(m^{2-\eta})$ for some $\eta>0$ on a DFA with $m$ transitions on a binary alphabet, then the Orthogonal Vectors problem with $d \in \Omega(\log N)$ can be solved in time $O(N^{2-\eta} \poly(d))$.
\end{theorem}

To prove Theorem \ref{thm:lower bound}, we adapt the reduction used by Equi et al.~\cite{EquiICALP19} to study the complexity of the \emph{pattern matching on labeled graphs} problem. 
The intuition is the following. 
Our new characterization of Wheeler languages states that we essentially need to find two distinct equally-labeled cycles in the minimum DFA for the language (in addition to checking some other properties on those cycles) in order to solve \textsc{WheelerLanguageDFA}. Given an instance of OV, we build a DFA (minimum for its language) having one (non-simple) cycle for each vector in the instance, such that the strings labeling two such cycles match if and only if the two corresponding vectors in the OV instance are orthogonal. As a result, a subquadratic-time solution of \textsc{WheelerLanguageDFA} on this DFA yields a subquadratic-time solution for the OV instance.

\section{Preliminaries}\label{sec:preliminaries}

\paragraph{\textbf{Strings}}

Let $\Sigma$ be a finite alphabet. A \emph{finite string} $\alpha\in \Sigma^*$ 
is a finite concatenation of characters from $\Sigma$. 
The notation $|\alpha|$ indicates the length of the string $\alpha$.
The symbol $\epsilon$ denotes the empty string. 
The notation $\alpha[i]$ denotes the $i$-th character from the beginning of $\alpha$, with indices starting from 1.
Letting $\alpha,\beta\in\Sigma^*$, $\alpha\cdot \beta$ (or simply $\alpha\beta$) denotes the concatenation of strings.
The notation $\alpha[i..j]$ denotes $\alpha[i]\cdot \alpha[i+1]\cdot\ \dots\ \cdot \alpha[j]$. 
An \emph{$\omega$-string} $\beta \in \Sigma^\omega$ (or \emph{infinite string} / \emph{string of infinite length}) is an infinite numerable concatenation of characters from $\Sigma$. 
In this paper, we work with \emph{left-infinite} $\omega$-strings, meaning that $\beta \in \Sigma^\omega$ is constructed from the empty string $\epsilon$ by prepending an infinite number of characters to it. In particular, the operation of appending a character $a\in\Sigma$ at the end of a $\omega$-string $\alpha \in \Sigma^\omega$ is well-defined and yields the $\omega$-string $\alpha a$. 
The notation $\alpha^\omega$, where $\alpha \in \Sigma^*$, denotes the concatenation of an infinite (numerable) number of copies of string $\alpha$.
The co-lexicographic (or co-lex) order $\prec$ of two strings $\alpha,\beta\in \Sigma^* \cup \Sigma^\omega$ is defined as follows. (i) $\epsilon \prec \alpha$ for every $\alpha \in \Sigma^+ \cup \Sigma^\omega$, and (ii) if $\alpha = \alpha'a$ and $\beta=\beta'b$ (with $a,b\in \Sigma$ and $\alpha',\beta'\in \Sigma^* \cup \Sigma^\omega$), $\alpha \prec \beta$ holds if and only if $(a\prec b)  \vee (a=b \wedge \alpha' \prec \beta')$.
In this paper, the symbols $\prec$ and $\preceq$ will be used to denote the total order on the alphabet and the co-lexicographic order between strings/$\omega$-strings. 
Notation $[N]$ indicates the set of integers $\{1,2,\dots, N\}$.

\paragraph{\textbf{DFAs, WDFAs, and Wheeler languages}}

In this paper, we work with deterministic finite state automata (DFAs): 

\begin{definition}[DFA]
A DFA $\mathcal A$ is a quintuple $(Q,\Sigma,\delta,s,F)$ where $Q$ is a finite set of states, $\Sigma$ is an alphabet set, $\delta:Q\times\Sigma\rightarrow Q$ is a transition function, $s(\in Q)$ is a source state, and $F(\subseteq Q)$ is a set of final states.
\end{definition}

For $u,v\in Q$ and $a\in\Sigma$ such that $\delta(u,a)=v$, we define $\lambda(u,v)=a$.
We extend the domain of the transition function to words $\alpha\in \Sigma^*$ as usual, i.e., for $a\in \Sigma$, $\alpha\in \Sigma^*$, and $u\in Q$: $\delta(u,a\cdot \alpha) = \delta(\delta(u,a),\alpha)$ and $\delta(u,\epsilon) = u$.

In this work, $n= |Q|$  denotes the number of states and $m = |\delta| = |\{(u,v,a)\in Q\times Q\times \Sigma :\delta(u,a)=v\}|$ the number of transitions of the input DFA. 

The notation $I_q$ indicates the set of words \emph{reaching} $q$ from the initial state:

\begin{definition}\label{def:I_q}
Let $\mathcal A = (Q, \Sigma, \delta, s, F) $ be a DFA.
If $u\in Q$, let $I_u$ be defined as: 
$$
		I_u =\{\alpha \in \Sigma^* : u = \delta(s, \alpha)\};
$$
\end{definition}

The \emph{language $\mathcal L(\mathcal A)$ recognized by $\mathcal A$} is defined as $\mathcal L(\mathcal A) = \cup_{u\in F}I_u$.

A classic result in language theory \cite{nerode1958linear} states that the minimum DFA --- denoted with $\mathcal A_{min}$ --- recognizing the language $\mathcal{L}(\mathcal A)$ 
of any DFA $\mathcal A$
is unique.  
The DFA $\mathcal A_{min}$  can be computed from $\mathcal A$ in $O(m\log n)$ time with a classic partition-refinement algorithm due to Hopcroft
\cite{hopcroft1971n}.

\emph{Wheeler automata} were introduced in \cite{gagie:tcs17:wheeler} as a generalization of prefix sorting from strings to labeled graphs. We consider the following particular case:

\begin{definition}[Wheeler DFA]\label{def:WDFA}
A \emph{Wheeler DFA} (WDFA for brevity) \cite{gagie:tcs17:wheeler}
$\mathcal A$ is a DFA for which there exists a \emph{total} order $<\ 
\subseteq Q\times Q$ (called \emph{Wheeler order}) satisfying the following three axioms:

\begin{itemize}
    \item[] \emph{(i)} $s < u$ for every $u\in Q-\{s\}$.
    \item[] \emph{(ii)} If $u' = \delta(u,a)$, $v' = \delta(v,b)$, and $a \prec b$, then $u'<v'$.
    \item[] \emph{(iii)} If $u' = \delta(u,a) \neq \delta(v,a) = v'$ and $u<v$, then $u' < v'$.
\end{itemize}
\end{definition}

The symbol $<$ will indicate both the total order of integers and the Wheeler order among the states of a Wheeler DFA.
The meaning of symbol $<$ will always be clear from the context.
Definition \ref{def:WDFA} defines the Wheeler order in terms of \emph{local} axioms. On DFAs, an equivalent \emph{global} definition is the following \cite{alanko:iac21:wheeler}:

\begin{definition}\label{def: <_A}
Let $u,v$ be two states of a DFA $\mathcal A$.
Let $u <_{\mathcal A} v$ if and only if $(\forall \alpha \in I_u) ( \forall \beta \in I_v) ~(\alpha \prec \beta)$.     
\end{definition}

\begin{lemma}[{\cite{alanko:iac21:wheeler}}]\label{lem_W order1}
    $\mathcal A$ is Wheeler if and only if $<_{\mathcal A}$ is total, if and only if $<_{\mathcal A}$ is the (unique) Wheeler order of $\mathcal A$.
\end{lemma}

In fact, when a Wheeler order exists for a DFA, this order is unique \cite{alanko:iac21:wheeler} (as opposed to the NFA case).
The class of languages recognized by Wheeler automata is of particular interest:

\begin{definition}[Wheeler language]
A regular language $\mathcal{L}$ is said to be \emph{Wheeler} if and only if there exists a Wheeler NFA $\mathcal A$ such that $\mathcal{L}=\mathcal{L}(\mathcal A)$, if and only if there exists a Wheeler DFA $\mathcal A'$ such that $\mathcal{L}=\mathcal{L}(\mathcal A')$.
\end{definition}

The equivalence between WNFAs and WDFAs was established in \cite{alanko:iac21:wheeler}. 
In the same paper \cite{alanko:iac21:wheeler}, the authors provided a \emph{Myhill-Nerode theorem} for Wheeler languages that is crucial for our results. 
Their result can be stated in terms of the minimum accepting DFA for $\mathcal L$. We first need the following definition:

\begin{definition}[Entanglement \cite{cotumaccio:arxiv22:partIpaper}]
Given a DFA $\mathcal A$, two distinct states $u \neq v$ of $\mathcal A$ are said to be \emph{entangled} if there exists a monotone infinite sequence $\alpha_1\prec\beta_1\prec\cdots\prec\alpha_i\prec\beta_i\prec\alpha_{i+1}\prec\beta_{i+1}\prec\cdots$ (or with reversed sign $\succ$) such that $\alpha_i\in I_u$ and $\beta_i\in I_v$ for every $i\ge 1$.
\label{def: entangle}
\end{definition}

The characterization of Wheeler languages of Alanko et al. \cite{alanko:iac21:wheeler} states that:

\begin{lemma}[\cite{alanko:iac21:wheeler}]
For a DFA $\mathcal A$, $\mathcal{L}(\mathcal A)$ is not Wheeler if and only if there exist entangled states $u$ and $v$ in its minimum DFA $\mathcal A_{min}$.
\label{lem: nonWheeler iff entangled}
\end{lemma}

Lemma \ref{lem: nonWheeler iff entangled} is at the core of our algorithm for recognizing Wheeler languages.

\vspace{-10pt}

\subsection{Infima and suprema strings}

Lemma \ref{lem_W order1} suggests that the Wheeler order can be defined by looking just at the lower and upper bounds of $I_u$ for each state $u\in Q$. Let us define:

\begin{definition}[Infimum and supremum \cite{kop:cpm23:working}]
For a DFA $\mathcal A=(Q,\Sigma,\delta,s,F)$, let $u\in Q$ be a state of $\mathcal A$. The infimum string $\inf I_u$ and supremum string $\sup I_u$ are the greatest lower bound and the least upper bound, respectively, of $I_u$:
\begin{align*}
\inf I_u &= \gamma\in\Sigma^*\cup\Sigma^\omega \mbox{ s.t. } (\forall \beta\in\Sigma^*\cup\Sigma^\omega \mbox{ s.t. }(\forall\alpha\in I_u~\beta\preceq\alpha)~ \beta\preceq\gamma) \\
\sup I_u &= \gamma\in\Sigma^*\cup\Sigma^\omega \mbox{ s.t. } (\forall \beta\in\Sigma^*\cup\Sigma^\omega \mbox{ s.t. } (\forall\alpha\in I_u~\alpha\preceq\beta)~ \gamma\preceq\beta) 
\end{align*}

\end{definition}

Kim et al. \cite{kop:cpm23:working} and Conte et al. \cite{conte2023computing} use the above definition to give yet another equivalent definition of Wheeler order:

\begin{lemma}[\cite{kop:cpm23:working,conte2023computing}]\label{lem_W order2}
Let $u,v$ be two states of a WDFA $\mathcal A$. 
Let $u < v$ if and only if $\sup{I_u}\preceq \inf{I_v}$. Then $<$ is the Wheeler order of $\mathcal A$.
\end{lemma}

Following Lemma \ref{lem_W order2}, it is convenient to represent each state $u\in Q$ as an \emph{open} interval $\mathcal{I}(u)=(\inf I_u,\sup I_u)$, i.e., the subset of $\Sigma^* \cup \Sigma^\omega$ containing all strings co-lexicographically strictly larger than $\inf I_u$ and strictly smaller than $\sup I_u$. 
Note that, for two states $u,v\in Q$,  
if $|I_u|,|I_v|>1$ then
$\mathcal{I}(u)\cap \mathcal{I}(v)=\emptyset$ if and only if $\sup I_u\preceq \inf I_v$ or $\sup I_v\preceq \inf I_u$. If $|I_u| = 1$ (analogously for $|I_v| = 1$), then $\mathcal I(u) = \emptyset$ so $\mathcal{I}(u)\cap \mathcal{I}(v)$ is always empty.

Following \cite{kop:cpm23:working}, in the rest of the paper the intervals $\mathcal{I}(u)=(\inf I_u,\sup I_u)$ are encoded as pairs of integers: the co-lexicographic ranks of $\inf I_u$ and $\sup I_u$ in $\{\inf I_u:u\in Q\}\cup\{\sup I_u:u\in Q\}$. Using this representation, the check $\mathcal{I}(u) \cap \mathcal{I}(v) \neq \emptyset$ can be trivially performed in constant time.
The authors of \cite{becker2023sorting} show that the relative co-lexicographic  ranks of all infima and suprema strings of a DFA can be computed efficiently:

\begin{lemma}[{\cite[Sec. 4]{becker2023sorting}}]\label{lem:pruning+sorting}
    Given a DFA $\mathcal{A}=(Q,\Sigma,\delta,s,F)$, we can sort the set $\{\inf I_u:u\in Q\}\cup\{\sup I_u:u\in Q\}$ co-lexicographically in $O(|\delta|\log |Q|)$ time.
\end{lemma}

We conclude this section by mentioning two useful properties of infima and suprema strings which will turn out useful later on in this work.

\begin{lemma}
    Let $u$ be a state of a DFA $\mathcal{A}$, and $\gamma\in\Sigma^*$ be a finite string. Then the following holds:
    \begin{enumerate}
        \item \label{obs: finite inf} If $\inf I_u$ ($\sup I_u$) is finite, then $\inf I_u\in I_u$ ($\sup I_u\in I_u$).
        \item \label{obs: suffix of inf} For any finite suffix $\alpha'$ of $\inf I_u$ or $\sup I_u$, there exists $\alpha\in I_u$ suffixed by $\alpha'$.
        \item \label{obs: singleton} $I_u$ 
is a singleton if and only if $\inf I_u=\sup I_u$. 
        \item \label{obs: alpha exists when gamme not eq infsup} If $\inf I_u\prec \gamma^\omega$, then there exists $\alpha\in I_u$ such that $\alpha\prec\gamma^\omega$; similarly, if $\gamma^\omega\prec\sup I_u$, then there exists $\alpha\in I_u$ such that $\gamma^\omega\prec\alpha$.
    \end{enumerate}
    \label{lem: observation}
\end{lemma}
\begin{proof}
    (1)-(3) See {\cite[Observation 8]{kop:cpm23:working}}.
    (4) Assume $\inf I_u\prec\gamma^\omega$.
    If $\inf I_u$ is finite, then $\inf I_u\in I_u$ by (\ref{obs: finite inf}) and the claim follows by setting $\alpha=\inf I_u$.
    Let us assume $\inf I_u$ has infinite length.
    Let $\alpha'$ be the shortest suffix of $\inf I_u$ such that $\alpha'\prec\gamma^\omega$ and $\alpha'$ is not a suffix of $\gamma^\omega$; note that $\alpha'$ is finite, otherwise $\alpha'=\inf I_u=\gamma^\omega$ by definition of $\prec$, which contradicts the assumption $\inf I_u\prec\gamma^\omega$. 
    Then by (\ref{obs: suffix of inf}), there exists $\alpha\in I_u$ suffixed by $\alpha'$. By definition of $\alpha'$, any string suffixed by $\alpha'$ is  smaller than $\gamma^\omega$, hence $\alpha\prec\gamma^\omega$. The case with $\gamma^\omega\prec\sup I_u$ is analogous.

\end{proof}

\section{Recognizing Wheeler Languages}

In this section, we present our algorithm to decide if the language accepted by a DFA $\mathcal{A}=(Q,\Sigma,\delta,s,F)$ is Wheeler.
Let $\mathcal{A}_{min}=(Q_{min},\Sigma,\delta_{min},s_{min},F_{min})$ be the minimum-size DFA accepting $\mathcal{L}(\mathcal{A})$.

\begin{definition}[Square automaton]
The \emph{square automaton} 
$\mathcal A_{min}^2=\mathcal{A}_{min}\times\mathcal{A}_{min}=(Q^2_{min}=Q_{min}\times Q_{min},\Sigma,\delta',(s_{min},s_{min}),F^2_{min}=F_{min}\times F_{min})$ 
is the automaton whose states are pairs of states of  $\mathcal{A}_{min}$ and whose transition function is defined as
    $\delta'((u,v),a)=(\delta_{min}(u,a),\delta_{min}(v,a))$ for $u,v\in Q_{min}$ and $a\in\Sigma$.
\end{definition}

We are ready to prove our new characterization of Wheeler languages.
The characterization states that $\mathcal A_{min}^2$ can be used to detect repeated cycles in $\mathcal A_{min}$, and that we can use this fact to check if $\mathcal L(\mathcal A_{min}) = \mathcal L(\mathcal A)$ is Wheeler:

\begin{theorem}
For a DFA $\mathcal{A}$, $\mathcal{L}(\mathcal{A})$ is not Wheeler if and only if $\mathcal A_{min}^2$ contains a cycle $(u_1,v_1) \rightarrow (u_2,v_2) \rightarrow \dots \rightarrow (u_k,v_k) \rightarrow (u_1,v_1)$ such that, for $1\le \forall i \le k$, the following hold: (i) $u_i\neq v_i$ and (ii) $\mathcal{I}(u_i)\cap\mathcal{I}(v_i)\ne\emptyset$.
\label{thm: nonWheeler iff cycle} 
\end{theorem}
\begin{proof}
$(\Leftarrow)$
Assume that $\mathcal{A}^2_{min}$ contains such a cycle $(u_1,v_1) \rightarrow (u_2,v_2) \rightarrow \dots \rightarrow (u_k,v_k) \rightarrow (u_1,v_1)$ where $k$ is the cycle length.
Then, by definition of $\mathcal{A}^2_{min}$ there exist cycles $u_1 \rightarrow u_2 \rightarrow \dots \rightarrow u_k \rightarrow u_1$ and $v_1 \rightarrow v_2 \rightarrow \dots \rightarrow v_k \rightarrow v_1$ in $\mathcal{A}_{min}$, both of which are labeled by the same string  $\gamma = \lambda(u_1,u_2)\cdots\lambda(u_{k-1},u_{k}) \lambda(u_{k},u_{1}) = \lambda(v_1,v_2)\cdots\lambda(v_{k-1},v_{k}) \lambda(v_{k},v_{1})$
of length $k$. 

Let $\gamma_1=\max\{\inf I_{u_1},\inf I_{v_1}\}$ and $\gamma_2=\min\{\sup I_{u_1},\sup I_{v_1}\}$. 
First, we claim $\gamma_1\prec\gamma_2$.
To see this, without loss of generality, assume $\inf I_{u_1}\preceq \inf I_{v_1}$. 
Observe $|I_{u_1}|,|I_{v_1}|>1$ because $u$ and $v$ are on two cycles (hence $I_{u_1}$ and $I_{v_1}$ contain an infinite number of strings). By Lemma~\ref{lem: observation}.\ref{obs: singleton} both $\inf I_{u_1}\prec \sup I_{u_1}$ and $\inf I_{v_1}\prec \sup I_{v_1}$ hold. Since $\mathcal{I}(u_1)\cap\mathcal{I}(v_1)\ne\emptyset$, $\inf I_{v_1} \prec \sup I_{u_1}$ also holds.
Then, $\inf I_{u_1} \preceq \inf I_{v_1} \prec \sup I_{u_1},\sup I_{v_1}$. Therefore $\gamma_1=\inf I_{v_1}\prec \min\{\sup I_{u_1},\sup I_{v_1}\} = \gamma_2$.

As a consequence, we can see that at least one of the following must hold: (i) $\gamma_1\prec\gamma^\omega$ and (ii) $\gamma^\omega\prec\gamma_2$; note that the complement of the case (i) is $\gamma^\omega\preceq \gamma_1$, which implies $\gamma^\omega\prec\gamma_2$ because $(\gamma^\omega\preceq)\gamma_1\prec\gamma_2$.
Therefore, by Lemma~\ref{lem: observation}.\ref{obs: alpha exists when gamme not eq infsup}, there must exist $\alpha\in I_{u_1}$ and $\beta \in I_{v_1}$ such that either $\alpha,\beta\prec \gamma^\omega$ or $\gamma^\omega\prec\alpha,\beta$ hold.

Note that it holds $\alpha\ne\beta$ since $\mathcal{A}_{min}$ is deterministic.
Without loss of generality, assume $\alpha\prec\beta$.
We consider the case $\alpha\prec\beta\prec\gamma^\omega$; the other case ($\gamma^\omega\prec \alpha\prec\beta$) is symmetric. 
Let $l$ be any integer such that $\max\{|\alpha|,|\beta|\}< l\cdot|\gamma|$.
Then we can see that, for every $d\ge 0$, the following three properties hold: (i) $\alpha(\gamma^{l})^d\prec\beta(\gamma^{l})^d\prec\alpha(\gamma^{l})^{d+1}\prec\beta(\gamma^{l})^{d+1}$, (ii) $\alpha(\gamma^{l})^d\in I_{u_1}$ (because $\alpha\in I_{u_1}$ and $\gamma$ labels a cycle from $u_1$, so $\delta_{min}(u_1,\gamma^k)=u_1$ for any integer $k\geq 0$) and, similarly, (iii) $\beta(\gamma^{l})^d\in I_{v_1}$. 
Properties (i-iii) imply that 
there is an infinite monotone nondecreasing sequence of strings alternating between $I_{u_1}$ and $I_{v_1}$, i.e., 
$u_1$ and $v_1$ are entangled (Definition~\ref{def: entangle}) and, by Lemma \ref{lem: nonWheeler iff entangled}, $\mathcal L(\mathcal A_{min}) = \mathcal L(\mathcal A)$ is not Wheeler. 

\medskip

$(\Rightarrow)$
Assume that $\mathcal L(\mathcal A_{min}) = \mathcal L(\mathcal A)$ is not Wheeler. By Lemma \ref{lem: nonWheeler iff entangled},  there exist entangled states $u_0 \neq v_0$ in $\mathcal{A}_{min}$ (in particular, $\mathcal I(u_0)\cap \mathcal I(v_0) \neq \emptyset$). 
Without loss of generality, we can assume that there is an infinite nondecreasing sequence $S_0=\alpha_1\prec\beta_1\prec\alpha_2\prec\beta_2\prec\cdots$ such that, for every $i\ge 1$, $\alpha_i\in I_{u_0}$ and $\beta_i\in I_{v_0}$ (the other case with the reversed sign is analogous).

Observe that, since the alphabet is finite, $S_0$  must ultimately (i.e., from a sufficiently large index $i$) contain strings $\alpha_i, \beta_i$ sharing the last character. We can therefore assume without loss of generality that all strings in  $S_0$ end with the same character $a$.
Then, there exist $u_1,v_1$ such that $\delta_{min}(u_1,a)=u_0$ and $\delta_{min}(v_1,a)=v_0$. 
Note that, by the determinism of $\mathcal A_{min}$, it must be $u_1 \neq v_1$. 
Moreover, we can choose two \emph{entangled} such $u_1,v_1$.
To see this, 
let $u_1^1,\dots, u_1^s$ and $v_1^1,\dots, v_1^r$ be the $s$ and $r$ predecessors of $u_0$ and $v_0$, respectively, such that $\delta_{min}(u_1^i,a)=u_0$ and $\delta_{min}(v_1^j,a)=v_0$ for all $1\leq i \leq s$ and $1\leq j \leq r$.
Assume for the purpose of contradiction that $u_1^i$ and $v_1^j$ are not entangled for all pairs $u_1^i, v_1^j$. Then, by definition of entanglement
any monotone sequence 
$\mu_1 \prec \mu_2 \prec \dots \in I_{u_1^i} \cup I_{v_1^j}$ ultimately ends up in \emph{just one} of the two sets: there exists $N\in \mathbb N$ such that either $\mu_N, \mu_{N+1}, \dots \in I_{u_1^i}$ or $\mu_N, \mu_{N+1}, \dots \in I_{v_1^j}$. 
Since this is true for \emph{any} pair $u_1^i$, $v_1^j$, any monotone sequence $\mu_1 \prec \mu_2 \prec \dots \in 
\bigcup_{i=1}^s I_{u_1^i} \cup \bigcup_{j=1}^r I_{v_1^j}$ ultimately ends up in \emph{either} (i) $\bigcup_{i=1}^s I_{u_1^i}$ \emph{or} (ii) $\bigcup_{j=1}^r I_{v_1^j}$. But then, this implies that sequence $S_0$ cannot exist: any monotone sequence $\mu_1a \prec \mu_2a \prec \dots \in I_{u_0} \cup I_{v_0}$ ultimately ends up in either (i) $I_{u_0}$  or  (ii) $I_{v_0}$.

Summing up, we found 
$u_1\neq v_1$ such that $\delta_{min}(u_1,a)=u_0$, $\delta_{min}(v_1,a)=v_0$, and
$u_1, v_1$ are
entangled (in particular, $\mathcal I(u_1)\cap \mathcal I(v_1) \neq \emptyset$).
We iterate this process for $k = |Q_{min}|^2$ times; this yields two paths $u_k \rightarrow u_{k-1} \rightarrow \dots \rightarrow u_0$ and $v_k \rightarrow v_{k-1} \rightarrow \dots \rightarrow v_0$ labeled with the same string of length $k$, with $u_i\neq v_i$ and $\mathcal I(u_i)\cap \mathcal I(v_i) \neq \emptyset$ for all $0\leq i \leq k$. But then, since we chose $k = |Q_{min}|^2$, by the pigeonhole principle there must exist two indices $j \leq i$ such that $(u_i,v_i) = (u_j,v_j)$. 
In particular, there exists $k'\leq k$ such that $u_i \rightarrow u_{i-1} \rightarrow \dots \rightarrow u_{i-k'+1} \rightarrow u_i$ and $v_i \rightarrow v_{i-1} \rightarrow \dots \rightarrow v_{i-k'+1} \rightarrow v_i$ are two cycles of the same length $k'$, labeled with the same string, such that $u_t\neq v_t$ and $\mathcal I(u_t)\cap \mathcal I(v_t) \neq \emptyset$ for all indices $i-k'+1 \leq t \leq i$. This yields our main claim.  \qed
\end{proof}

\section{The algorithm}

Theorem \ref{thm: nonWheeler iff cycle} immediately gives a quadratic algorithm for \textsc{WheelerLanguageDFA}:
\begin{enumerate}
    \item Compute $\mathcal{A}_{min}=(Q_{min},\Sigma,\delta_{min},s_{min},F_{min})$ by Hopcroft's algorithm \cite{hopcroft1971n}.
    \item On $\mathcal A_{min}$, compute intervals $\mathcal I(u)$ for each $u\in Q_{min}$, using \footnote{
    In Appendix \ref{app:assumptions partition refinement} we discuss more in detail how to apply \cite[Sec. 4]{becker2023sorting} on $\mathcal A_{min}$.} \cite[Sec. 4]{becker2023sorting}.
    \item Compute $\mathcal A_{min}^2$.
    \item Remove from $\mathcal A_{min}^2$ all states $(u,v)$ 
    (and incident transitions) such that either $u=v$ or $\mathcal I(u) \cap \mathcal I(v) = \emptyset$. Let $\mathcal {\hat A}_{min}^2$ be the resulting pruned automaton. 
    \item Test acyclicity of $\mathcal {\hat A}_{min}^2$. 
    If $\mathcal {\hat A}_{min}^2$ is acyclic, 
    return \texttt{"$\mathcal{\mathcal L(\mathcal A)}$ is Wheeler"}. Otherwise, return \texttt{"$\mathcal{\mathcal L(\mathcal A)}$ is not Wheeler"}.
\end{enumerate}

Since, by its definition,  $\mathcal A_{min}$ cannot be larger than  $\mathcal A$, in the rest of the paper we will for simplicity assume that $\mathcal A_{min}$ has $n$ nodes and $m$ transitions. Steps (1) and (2) run in $O(m\log n)$ time. 
Note that, for each transition $\delta_{min}(u,a)=u'$ and for each node $v\neq u$, by the determinism of $\mathcal A_{min}$ there exists at most one transition $\delta_{min}(v,a)=v'$ labeled with $a$ and originating in $v$; such a pair of transitions define one transition of $\mathcal A_{min}^2$. It follows that the number of transitions (thus the size) of $\mathcal A_{min}^2$ is $O(mn)$, therefore steps (3-5) run in $O(mn)$ time (acyclicity can be tested in $O(|\mathcal A_{min}^2|)$ time using, for example, Kahn's topological sorting algorithm). Overall, the algorithm runs in $O(mn)$ time.

\subsection{A parameterized algorithm}
\label{sec:param-algo}
Our algorithm can be optimized by observing that we can directly build $\mathcal {\hat A}_{min}^2$, and that this automaton could be much smaller than $\mathcal A_{min}^2$. 
For example observe that, if $\mathcal A_{min}$ is Wheeler, then $\mathcal I(u) \cap \mathcal I(v) = \emptyset$ for all states $u\neq v$ of $\mathcal A_{min}$ (see Definition \ref{def: <_A} and Lemma \ref{lem_W order1}), so $\mathcal {\hat A}_{min}^2$ is empty. 
As a matter of fact, we show that the size of $\mathcal {\hat A}_{min}^2$ depends on the \emph{width} of the (partial \cite{cotumaccio:soda21:psortable}) order $<_{\mathcal A_{min}}$, i.e., the size of the largest antichain:

\begin{definition}[\cite{cotumaccio:soda21:psortable}]\label{def:colex width}
   The \emph{co-lex width} $width(\mathcal A)$ of a DFA $\mathcal A$ is the width of the order $<_{\mathcal A}$ defined in Definition \ref{def: <_A}.
\end{definition}

The co-lex width is an important measure parameterizing problems such as pattern matching on graphs and compression of labeled graphs \cite{cotumaccio:arxiv22:partIpaper,cotumaccio:soda21:psortable}. 
Note that $width(\mathcal A_{min})=1$ if and only if $\mathcal A_{min}$ is Wheeler. In Appendix \ref{appendix:parameterized} we prove:

\begin{lemma}\label{lem:size of pruned square A}
    Let $p = width(\mathcal A_{min})$. Then, $\mathcal {\hat A}_{min}^2$ has at most $2n(p-1)$ states and at most $2m(p-1)$ transitions and can be built from $\mathcal A_{min}$ in $O(mp)$ time.
\end{lemma}

The intuition behind Lemma \ref{lem:size of pruned square A} is that $\mathcal {\hat A}_{min}^2$ contains only states $(u,v)$ such that $\mathcal I(u) \cap \mathcal I(v) \neq \emptyset$. By Lemma \ref{lem_W order2}, this holds if and only if $u$ and $v$ are incomparable by the order  $<_{\mathcal A_{min}}$. Since the width of this order is (by definition) $p$, the bounds follow easily. 
To build $\mathcal {\hat A}_{min}^2$, we sort the states of $\mathcal {\hat A}_{min}$ by the strings $\inf I_u$ and observe that incomparable states are adjacent in this order. It follows that we can easily build $\mathcal {\hat A}_{min}^2$ in time proportional to its size, $O(mp)$. The details of this algorithm can be found in Appendix \ref{appendix:parameterized}.
Theorem \ref{theorem: main result} follows.

\begin{figure}[b!]
    \centering
    \begin{tabular}{cc}
 	\includegraphics[width=0.44\textwidth, trim={5.5mm 5.5mm 5.0mm 5.5mm}, clip]{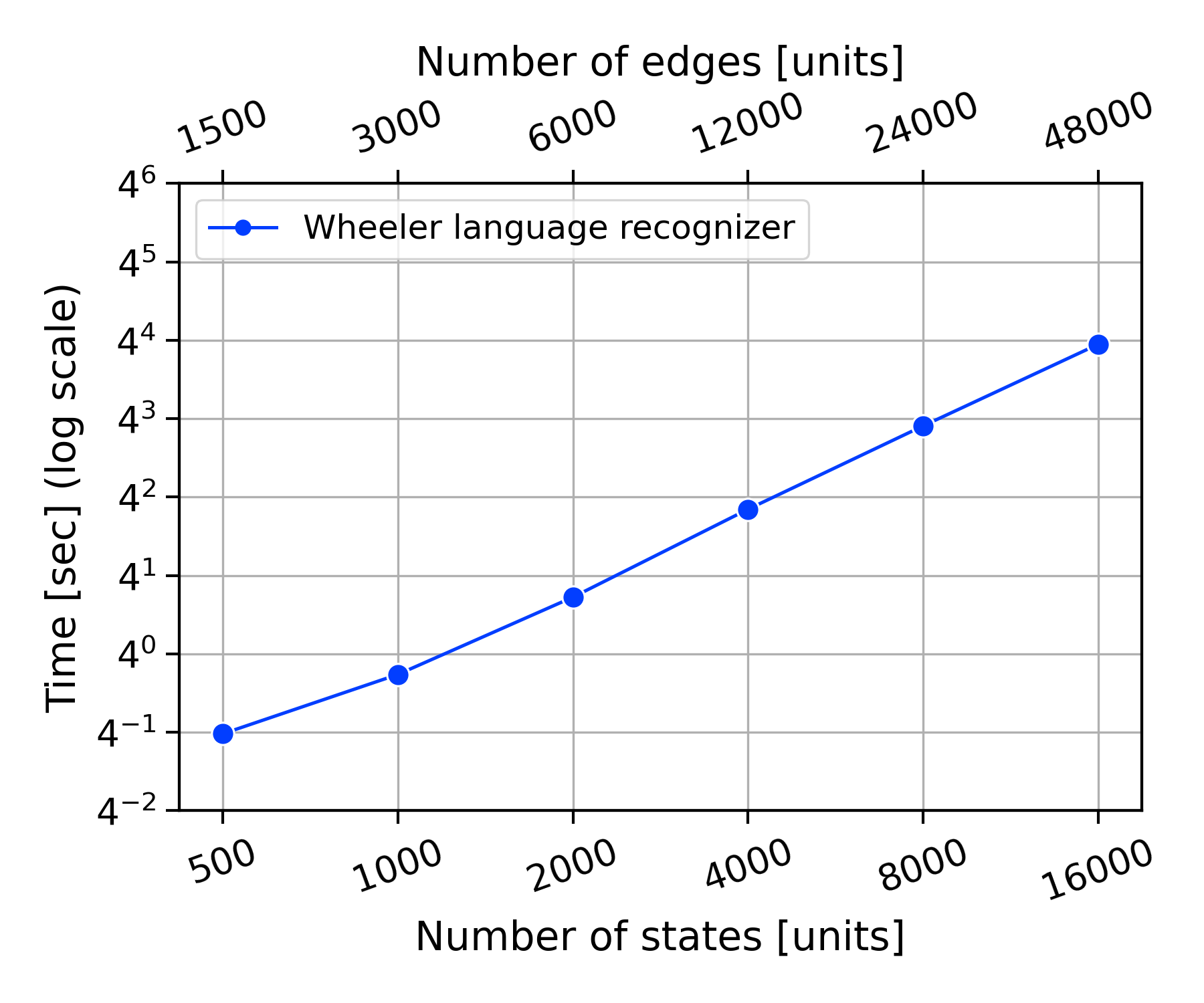}
    \hspace{3mm}&
    \includegraphics[width=0.44\textwidth, trim={5.5mm 5.5mm 5.0mm 5.5mm}, clip]{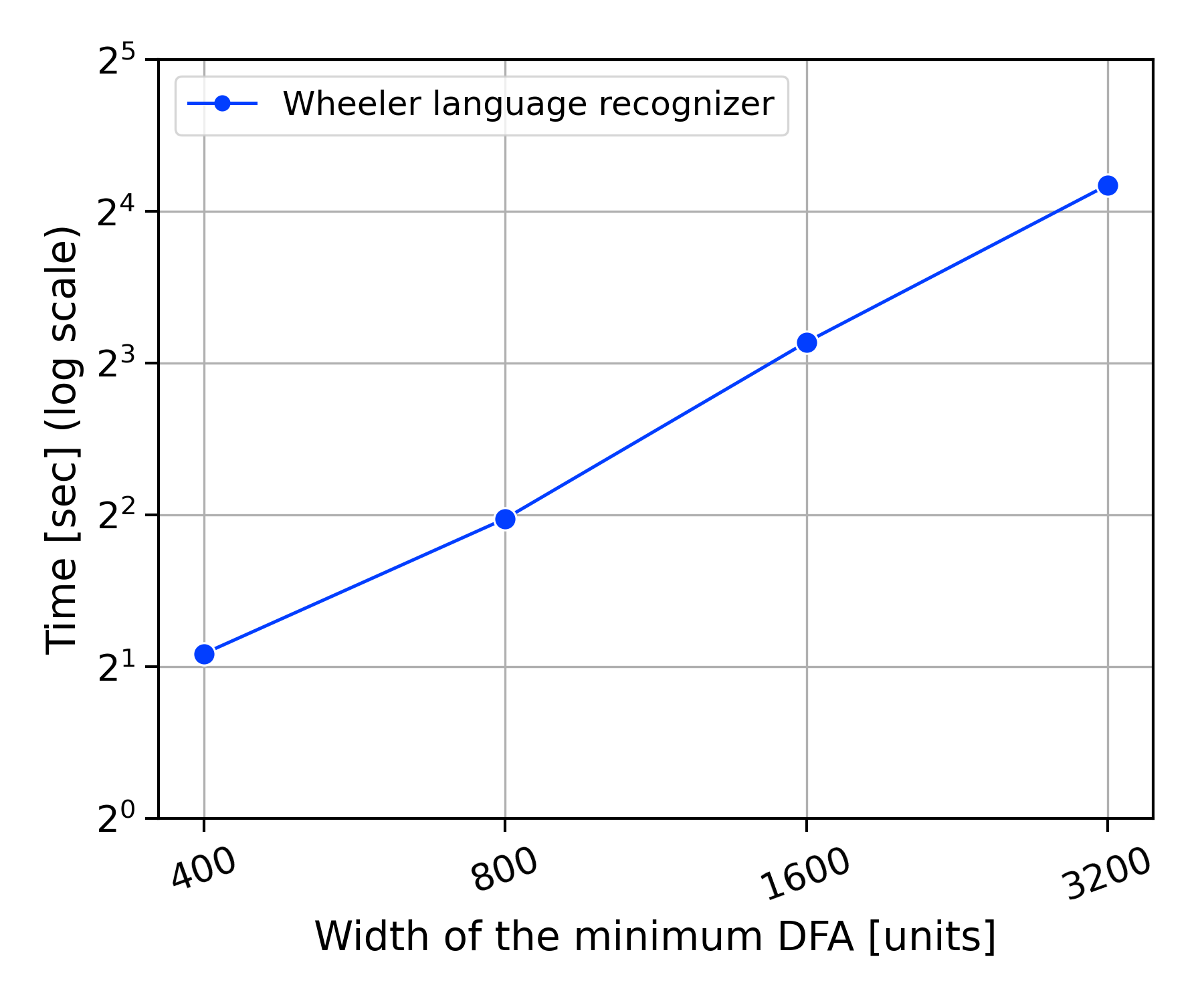}\\
    (a)&(b)\\
    \end{tabular}
    \caption{Wall clock time for our algorithm on different random DFA datasets (a) different $n$, $m=3n$, and $p$ is similar to $n$; (b) different $p$ with fixed $m$.
    }
    \label{fig:exp1}
\end{figure}

\subsubsection*{Implementation}\label{sec:implementation}

We implemented the algorithm of Theorem \ref{theorem: main result}. The code is available at \url{http://github.com/regindex/Wheeler-language-recognizer}. It takes in input either a regular expression or a DFA and checks if the recognized language is Wheeler. We tested our algorithm on two random DFA datasets: (i) one with different combinations of number of states and transitions where $n=\{500\cdot2^i: i=0,\ldots,5\}$ and $m=3n$ to show the quadratic running time, and (ii) the other with a fixed number of transitions, $m=16\cdot 10^3$, and different widths $p=\{400,800,1600,3200\}$ of the minimum DFAs to show the running time with respect to $p$. Our experiments were run on a server with Intel(R) Xeon(R) W-2245498 CPU @ 3.90GHz with 8 cores and 128 gigabytes of RAM running Ubuntu 18.04 LTS 64-bit. As expected, our experimental results show that the running time grows linearly in $mp$. It is worth noting that, on our input instances in the first dataset, $p$ is roughly similar to $n$, and we double $n$ at each step, the running time shows a quadratic growth (Fig. \ref{fig:exp1}(a)). On the other hand, when the number of transitions is fixed, the running time grows linearly to the width $p$ of the minimum DFA (Fig. \ref{fig:exp1}(b)). This can be measured from the slopes of the fitted lines on the log-log plots, which are 2.03 and 1.04, respectively.

\section{A matching conditional lower bound}
In this section, we show that an algorithm for \textsc{WheelerLanguageDFA} with running time $O(m^{2-\eta})$, yields an algorithm for the Orthogonal Vectors problem (see Definition~\ref{def:OV}) with running time $O(N^{2-\eta}\poly(d))$, thus contradicting SETH. This is our second main theorem (Theorem~\ref{thm:lower bound}) formulated
in the introduction. We prove this theorem using Theorem~\ref{thm: nonWheeler iff cycle} and the following proposition, which reduces an instance of the OV problem with two sets of $N$ $d$-dimensional vectors each into an instance of our problem with a minimum DFA of size $\Theta(Nd)$.
\begin{proposition}\label{prop: reduction}
    For an instance of the OV problem, we can in $O(N(d+\log N))$ time construct a DFA $\mathcal A$ 
    with $m\in O(N(d+\log N))$ edges
    that is minimum for its language $\mathcal L(\mathcal A)$ such that the OV instance is a YES-instance if and only if $\mathcal{A}^2$ contains a cycle $(u_1,v_1) \rightarrow (u_2,v_2) \rightarrow \dots \rightarrow (u_k,v_k) \rightarrow (u_1,v_1)$ such that, for $1\le \forall i \le k$, the following hold: (i) $u_i\neq v_i$ and (ii) $\mathcal{I}(u_i)\cap\mathcal{I}(v_i)\ne\emptyset$.
\end{proposition}
Once this proposition is established, we can take an OV instance with sets of size $N$ containing vectors of dimension $d \in \omega(\log N)$ and construct the DFA $\mathcal A$ of size $\Theta(m) = \Theta(N(d+\log N)) = \Theta(Nd)$. Now assume that we can solve \textsc{WheelerLanguageDFA} in $O(m^{2-\eta})$ on $\mathcal A$. Using Theorem~\ref{thm: nonWheeler iff cycle} and Proposition~\ref{prop: reduction}, we can thus solve the OV instance in $O((Nd)^{2-\eta}) = O(N^{2-\eta}\poly(d))$ time, as the OV instance is a YES instance if and only if the language recognized by $\mathcal A$ is not Wheeler. This shows Theorem~\ref{thm:lower bound}.
The rest of this section is dedicated to illustrate how we prove Proposition~\ref{prop: reduction}. The details are deferred to Appendix~\ref{appendix: reduction}. 

\subsubsection*{Construction of $\mathcal A$}
For a given instance $A=\{a_1,\ldots, a_N\}$ and $B=\{b_1,\ldots, b_N\}$ of the OV problem, we build a DFA $\mathcal{A}= (Q,\Sigma,\delta,s,F)$ with the properties in Proposition~\ref{prop: reduction} by adapting a technique of Equi et al.~\cite{EquiICALP19}. We first notice that we can, w.l.o.g., make the following assumptions on the OV instance: (1) The vectors in $A$ are distinct, (2) $N$ is a power of two, say $N= 2^\ell$. 
We describe the construction of $\mathcal A$ based on a small example, while the general description is deferred to Appendix~\ref{appendix: reduction}. We let $\Sigma=\{0,1,\#\}$. Later we show how to reduce the alphabet's size to 2.  

\begin{figure}[ht]
  \centering{
    \resizebox{0.9\columnwidth}{!}{
\begin{tikzpicture}[
          scale=.6,
          ->,
          >=stealth',
          shorten >=1pt,
          auto,
          semithick,
          every node/.style={minimum size=8mm}
        ]

        \begin{scope}[xshift=0cm]
        	\node at (4, 5.85) {\LARGE $C^A_1:\quad a_1 = 110$};
            \node at (4, 0)        {\large $110\ \{0,1\}\{0,1\}\ \# $};
        	\node[circle, draw] (A)   at (0,0)         {$\hat a_1$};
        	\node[circle, draw] (B)   at (2, -3.464)   {$a_1^1$};
        	\node[circle, draw] (C)   at (6, -3.464)   {$a_1^2$};
        	\node[circle, draw] (D)   at (8, 0)        {$a_1^3$};
        	\node[circle, draw] (E)   at (6, 3.464)    {$q_1^1$};
        	\node[circle, draw] (F)   at (2, 3.464)    {$q_1^2$};
        	\path (A) edge [sloped] node {$1$} (B)
	              (B) edge [sloped] node {$1$} (C)
	              (C) edge [sloped] node {$0$} (D)
	              (D) edge[bend right] [sloped] node {$0$} (E)
                    (D) edge[bend left]  [sloped] node {$1$} (E)
                    (E) edge[bend right] [sloped, below] node {$0$} (F)
                    (E) edge[bend left]  [sloped, below] node {$1$} (F)
	              (F) edge [sloped] node {$\#$} (A);
        \end{scope}

        \begin{scope}[xshift=11cm]
            \node at (4, 5.5) {\LARGE $C^A_2:\quad \textcolor{red}{a_2 = 100}$};
            \node at (4, 0)        {\large $\textcolor{red}{100}\ \{0,\textcolor{red}{1}\}\{\textcolor{red}{0},1\}\textcolor{red}{\ \#}$};
        	\node[circle, draw] (A)   at (0,0)         {$\hat a_2$};
        	\node[circle, draw] (B)   at (2, -3.464)   {$a_2^1$};
        	\node[circle, draw] (C)   at (6, -3.464)   {$a_2^2$};
        	\node[circle, draw] (D)   at (8, 0)        {$a_2^3$};
        	\node[circle, draw] (E)   at (6, 3.464)    {$q_2^1$};
        	\node[circle, draw] (F)   at (2, 3.464)    {$q_2^2$};
        	\path (A) edge[red, line width=1mm] [sloped] node {$1$} (B)
	              (B) edge[red, line width=1mm] [sloped] node {$0$} (C)
	              (C) edge[red, line width=1mm] [sloped] node {$0$} (D)
	              (D) edge[bend right] [sloped] node {$0$} (E)
                    (D) edge[bend left, red, line width=1mm] [sloped]  node {$1$} (E)
                    (E) edge[bend right, red, line width=1mm]  [sloped, below] node {$0$} (F)
                    (E) edge[bend left] [sloped, below]  node {$1$} (F)
                    (F) edge[red, line width=1mm] [sloped] node {$\#$} (A);
        \end{scope}

        \begin{scope}[xshift=22cm]
            \node at (4, 5.5) {\LARGE $C^A_3:\quad a_3 = 111$};
            \node at (4, 0)        {\large $111 \  \{0,1\}\{0,1\}\ \#$};
        	\node[circle, draw] (A)   at (0,0)         {$\hat a_3$};
        	\node[circle, draw] (B)   at (2, -3.464)   {$a_3^1$};
        	\node[circle, draw] (C)   at (6, -3.464)   {$a_3^2$};
        	\node[circle, draw] (D)   at (8, 0)        {$a_3^3$};
        	\node[circle, draw] (E)   at (6, 3.464)    {$q_3^1$};
        	\node[circle, draw] (F)   at (2, 3.464)    {$q_3^2$};
        	\path (A) edge [sloped] node {$1$} (B)
	              (B) edge [sloped] node {$1$} (C)
	              (C) edge [sloped] node {$1$} (D)
	              (D) edge[bend right] [sloped] node {$0$} (E)
                    (D) edge[bend left] [sloped]  node {$1$} (E)
                    (E) edge[bend right] [sloped, below] node {$0$} (F)
                    (E) edge[bend left] [sloped, below]  node {$1$} (F)
                    (F) edge [sloped] node {$\#$} (A);
        \end{scope}

        \begin{scope}[xshift=33cm]
            \node at (4, 5.5) {\LARGE $C^A_4:\quad a_4 = 011$};
            \node at (4, 0)        {\large $011 \  \{0,1\}\{0,1\}\ \#$};
        	\node[circle, draw] (A)   at (0,0)         {$\hat a_4$};
        	\node[circle, draw] (B)   at (2, -3.464)   {$a_4^1$};
        	\node[circle, draw] (C)   at (6, -3.464)   {$a_4^2$};
        	\node[circle, draw] (D)   at (8, 0)        {$a_4^3$};
        	\node[circle, draw] (E)   at (6, 3.464)    {$q_4^1$};
        	\node[circle, draw] (F)   at (2, 3.464)    {$q_4^2$};
        	\path (A) edge [sloped] node {$0$} (B)
	              (B) edge [sloped] node {$1$} (C)
	              (C) edge [sloped] node {$1$} (D)
	              (D) edge[bend right] [sloped] node {$0$} (E)
                   (D) edge[bend left] [sloped]  node {$1$} (E)
                   (E) edge[bend right] [sloped, below] node {$0$} (F)
                   (E) edge[bend left] [sloped, below]  node {$1$} (F)
                   (F) edge [sloped] node {$\#$} (A);
        \end{scope}

        \begin{scope}[xshift=0cm, yshift=-11.5cm]
            \node at (4, 5.5) {\LARGE $C^B_1:\quad b_1 = 101$};
            \node at (4, 0)        {\large $0\{0,1\}0 \  00 \ \#$};
        	\node[circle, draw] (A)   at (0,0)         {$\hat b_1$};
        	\node[circle, draw] (B)   at (2, -3.464)   {$b_1^1$};
        	\node[circle, draw] (C)   at (6, -3.464)   {$b_1^2$};
        	\node[circle, draw] (D)   at (8, 0)        {$b_1^3$};
        	\node[circle, draw] (E)   at (6, 3.464)    {$p_1^1$};
        	\node[circle, draw] (F)   at (2, 3.464)    {$p_1^2$};
        	\path (A) edge [sloped] node {$0$} (B)
	              (B) edge[bend right] [sloped] node {$1$} (C)
                  (B) edge[bend left]  [sloped] node {$0$} (C)
	              (C) edge [sloped] node {$0$} (D)
	              (D) edge [sloped] node {$0$} (E)
                    (E) edge [sloped] node {$0$} (F)
                    (F) edge [sloped] node {$\#$} (A);
        \end{scope}

        \begin{scope}[xshift=11cm, yshift=-11.5cm]
            \node at (4, 5.5) {\LARGE $C^B_2:\quad b_2 = 101$};
            \node at (4, 0)        {\large $0\{0,1\}0 \  01 \ \#$};
        	\node[circle, draw] (A)   at (0,0)         {$\hat b_2$};
        	\node[circle, draw] (B)   at (2, -3.464)   {$b_2^1$};
        	\node[circle, draw] (C)   at (6, -3.464)   {$b_2^2$};
        	\node[circle, draw] (D)   at (8, 0)        {$b_2^3$};
        	\node[circle, draw] (E)   at (6, 3.464)    {$p_2^1$};
        	\node[circle, draw] (F)   at (2, 3.464)    {$p_2^2$};
        	\path (A) edge [sloped] node {$0$} (B)
	              (B) edge[bend right] [sloped] node {$1$} (C)
                  (B) edge[bend left]  [sloped] node {$0$} (C)
	              (C) edge [sloped] node {$0$} (D)
	              (D) edge [sloped] node {$0$} (E)
                    (E) edge [sloped] node {$1$} (F)
                    (F) edge [sloped] node {$\#$} (A);
        \end{scope}

        \begin{scope}[xshift=22cm, yshift=-11.5cm]
            \node at (4, 5.5) {\LARGE $C^B_3:\quad \textcolor{red}{b_3 = 010}$};
            \node at (4, 0)        {\large $\{0,\textcolor{red}{1}\}\textcolor{red}{0}\{\textcolor{red}{0},1\} \textcolor{red}{\ 10\ \#}$};
        	\node[circle, draw] (A)   at (0,0)         {$\hat b_3$};
        	\node[circle, draw] (B)   at (2, -3.464)   {$b_3^1$};
        	\node[circle, draw] (C)   at (6, -3.464)   {$b_3^2$};
        	\node[circle, draw] (D)   at (8, 0)        {$b_3^3$};
        	\node[circle, draw] (E)   at (6, 3.464)    {$p_3^1$};
        	\node[circle, draw] (F)   at (2, 3.464)    {$p_3^2$};
        	\path (A) edge[bend right, red, line width=1mm] [sloped] node {$1$} (B)
                  (A) edge[bend left]  [sloped] node {$0$} (B)
                  (B) edge[red, line width=1mm] [sloped] node {$0$} (C)
	              (C) edge[bend right] [sloped] node {$1$} (D)
                  (C) edge[bend left, red, line width=1mm]  [sloped] node {$0$} (D)
	              (D) edge[red, line width=1mm] [sloped] node {$1$} (E)
                    (E) edge[red, line width=1mm] [sloped] node {$0$} (F)
                    (F) edge[red, line width=1mm] [sloped] node {$\#$} (A);
        \end{scope}

        \begin{scope}[xshift=33cm, yshift=-11.5cm]
            \node at (4, 5.5) {\LARGE $C^B_4:\quad b_4 = 111$};
            \node at (4, 0)        {\large $000 \  11 \ \#$};
        	\node[circle, draw] (A)   at (0,0)         {$\hat b_4$};
        	\node[circle, draw] (B)   at (2, -3.464)   {$b_4^1$};
        	\node[circle, draw] (C)   at (6, -3.464)   {$b_4^2$};
        	\node[circle, draw] (D)   at (8, 0)        {$b_4^3$};
        	\node[circle, draw] (E)   at (6, 3.464)    {$p_4^1$};
        	\node[circle, draw] (F)   at (2, 3.464)    {$p_4^2$};
        	\path (A) edge [sloped] node {$0$} (B)
                  (B) edge [sloped] node {$0$} (C)
	              (C) edge [sloped] node {$0$} (D)
	              (D) edge [sloped] node {$1$} (E)
                    (E) edge [sloped] node {$1$} (F)
                    (F) edge [sloped] node {$\#$} (A);
        \end{scope}

      \end{tikzpicture}
    }
  }
\caption{Illustration of the cycles generated for the bit vectors in the example $A=\{110, 100, 111, 011\}$ and $B=\{101, 101, 010, 111\}$. The two cycles $C^A_2$ and $C^B_3$ generate a match, i.e., can read the same string, as the vectors $a_2$ and $b_3$ are orthogonal.}
\label{fig:reduction example}
\end{figure}
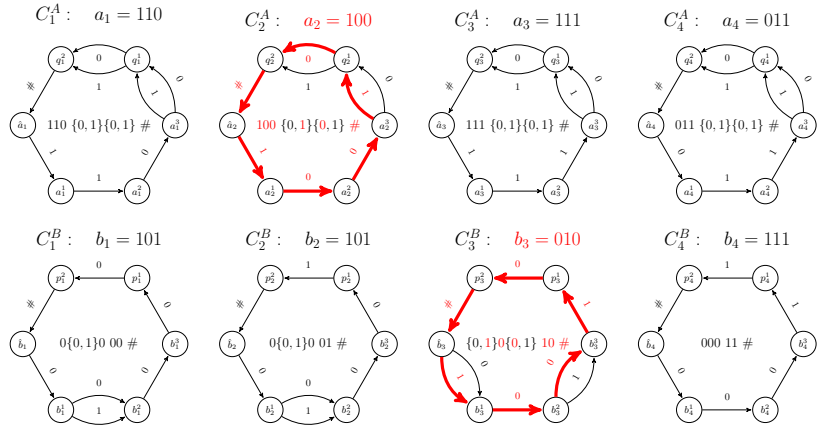

Let $A=\{110, 100, 111, 011\}$ and $B=\{101, 101, 010, 111\}$ be the given instance of the OV problem. Thus $N=4$ and $\ell=\log_2N = 2$. Notice that the only pair of orthogonal vectors in $A$ and $B$ are $a_2 = 100$ and $b_3 = 010$. 
In our DFA $\mathcal A$ we build (non-simple) cycles $C_i^A$ and $C_j^B$ for every vector $a_i$ and $b_j$ in $A$ and $B$ respectively.
As an example, for $a_2$ we build the cycle $C_2^A$ labeled with $100\{0,1\}\{0,1\}\#$, i.e.\ the bit string $100$ of $a_2$, followed by a sub-graph recognizing any bit string of length $\ell=2$, followed by $\#$. For $b_3$ we build the (non-simple) cycle $C_3^B$ labeled with $\{0,1\}0\{0,1\} 10 \#$, i.e., the bits 0 of $b_3$ are converted to a sub-graph recognizing both 0 and 1, and the bits 1 of $b_3$ are converted to an edge recognizing 0; this subgraph is followed by a path of length $\ell$ spelling 10, which is the 3rd smallest among the length-$\ell$ binary strings (i.e. the identifier for $C_3^B$ to prevent it from any match with $C_j^B$ for $j\ne 3$ while allowing matches with $C_i^A$'s), which is followed by an edge labeled with $\#$. Notice that these two cycles indeed generate a match (underlined characters indicate the match): $\underline{100}\{0,\underline 1\}\{\underline 0,1\}\underline{\#}$ and $\{0,\underline 1\}\underline 0\{\underline 0,1\}  \underline{10}\underline{\#}$. It is not hard to see that a cycle $C_i^A$ and $C_j^B$ built in this way will match if and only if the two corresponding vectors are orthogonal. Characters $\#$ are introduced to synchronize the match (otherwise, other rotations of the cycles could match). The subgraphs between the part corresponding to the input vectors and the character $\#$ are introduced to avoid that two distinct cycles $C^B_i$ and $C^B_j$ ($i\neq j$) generate a match. Note that, since we assume that $A$ contains distinct vectors, distinct cycles $C^A_i$ and $C^A_j$ ($i\neq j$) will never generate a match. 

The remaining details of the reduction ensure that (1)~the graph is a connected DFA, (2)~corresponding nodes (i.e.\ same distance from $\#$ in the cycles) $u,v$ in any pair of cycles $C^A_i$ and $C^B_j$ that correspond to orthogonal vectors $a_i$ and $b_j$, respectively, have a non-empty co-lexicographic intersection $\mathcal I(u) \cap \mathcal I(v)$, (3)~the DFA is indeed minimum for its recognized language, and (4)~the alphabet can be reduced to $\{0,1\}$ by an opportune mapping.

\begin{figure}[t]
  \centering{
    \resizebox{0.9\columnwidth}{!}{
      \begin{tikzpicture}[
          scale=.6,
          ->,
          >=stealth',
          shorten >=1pt,
          auto,
          semithick,
          every node/.style={minimum size=7mm}
        ]
            \node at (-2.0, 0) {\Huge $\mathcal A$:};

        	\node[circle, draw] (s)   at (0.5,0)  {$s$};
        	\node[circle, draw] (b)   at (2.5,-4)  {};
        	\node[circle, draw] (a)   at (2.5,4)  {};
        	\path (s) edge [sloped] node {$1$} (b);
        	\path (s) edge [sloped] node {$0$} (a);
            \node[circle, draw] (a1)   at (4.5,6)  {};
        	\node[circle, draw] (a2)   at (4.5,2)  {};
        	\path (a) edge [sloped] node {$0$} (a1);
        	\path (a) edge [sloped] node {$1$} (a2);
            \node[circle, draw] (b1)   at (4.5,-2)  {};
        	\node[circle, draw] (b2)   at (4.5,-6)  {};
        	\path (b) edge [sloped] node {$0$} (b1);
        	\path (b) edge [sloped] node {$1$} (b2);

            \node (d) at (6, 3) {$\ldots$};
            \node (d) at (6, -5) {$\ldots$};

            \path [draw = black, rounded corners, inner sep=100pt, dotted]
            (-0.75, 1) -- (2.5, 7.25) -- (10.15, 8.5) -- (10.15, -8.5) -- (2.5, -7.25) -- (-0.75, -1) -- cycle ;
            \node  (empty)    at (1.0, -6.5)  {\Large $V^{out}$};

            \node[circle, draw] (xN)   at (9, 2)  {$x_{N}$};
        	\node at (9, 4.7) {$\vdots$};
        	\node[circle, draw] (x1)   at (9, 7)  {$x_{1}$};
            \node[circle, draw] (y1)   at (9, -1.5)  {$y_{1}$};
        	\node[circle, draw] (y2)   at (9, -3.5)  {$y_{2}$};
        	\node at (9, -5) {$\vdots$};
        	\node[circle, draw] (yN)   at (9, -7)  {$y_{N}$};

            \node[circle, draw] (y1p)   at (12, -1.5)  {$y_{1}'$};
        	\node[circle, draw] (y2p)   at (12, -3.5)  {$y_{2}'$};
        	\node[circle, draw] (yNp)   at (12, -7)  {$y_{N}'$};
            \node[circle, draw] (xNp)   at (12, 1)  {$x_{N}''$};
            \node[circle, draw] (xNpp)   at (12, 3)  {$x_{N}'$};
        	\node[circle, draw] (x1p)   at (12, 6)  {$x_{1}''$};
            \node[circle, draw] (x1pp)   at (12, 8)  {$x_{1}'$};

            \node[circle, draw] (a1p)   at (15, 7)  {$\hat a_{1}'$};
            \node[circle, draw] (aNp)   at (15, 2)  {$\hat a_{N}'$};
            \node[circle, draw] (b1p)   at (15, -1.5)  {$\hat b_{1}'$};
        	\node[circle, draw] (b2p)   at (15, -3.5)  {$\hat b_{2}'$};
        	\node[circle, draw] (bNp)   at (15, -7)  {$\hat b_{N}'$};

            \begin{scope}[xshift=3cm]
                \node[circle, draw] (a1)   at (15, 7)  {$\hat a_{1}$};
            	\node[circle, draw] (aN)   at (15, 2)  {$\hat a_{N}$};
                \node[circle, draw] (b1)   at (15, -1.5)  {$\hat b_{1}$};
            	\node[circle, draw] (b2)   at (15, -3.5)  {$\hat b_{2}$};
            	\node[circle, draw] (bN)   at (15, -7)  {$\hat b_{N}$};
    
                \path   (y1) edge [sloped] node {$0$} (y1p)
                        (y1p) edge [sloped] node {$1$} (b1p);
                \path   (y2) edge [sloped] node {$0$} (y2p)
                        (y2p) edge [sloped] node {$1$} (b2p);
                \path   (yN) edge [sloped] node {$0$} (yNp)
                        (yNp) edge [sloped] node {$1$} (bNp);
                \path   (x1) edge [sloped] node {$0$} (x1p)
                        (x1) edge [sloped] node {$1$} (x1pp)
                        (x1p) edge [sloped] node {$0$} (a1p)
                        (x1pp) edge [sloped] node {$1$} (a1p);
                \path   (xN) edge [sloped] node {$0$} (xNp)
                        (xN) edge [sloped] node {$1$} (xNpp)
                        (xNp) edge [sloped] node {$0$} (aNp)
                        (xNpp) edge [sloped] node {$1$} (aNp);
                \path   (a1p) edge [sloped] node {$0$} (a1)
                        (aNp) edge [sloped] node {$0$} (aN)
                        (b1p) edge [sloped] node {$0$} (b1)
                        (b2p) edge [sloped] node {$0$} (b2)
                        (bNp) edge [sloped] node {$0$} (bN);
    
                \path [draw = black, rounded corners, inner sep=100pt, dotted]
                (7.35, 9.0) -- (13.65, 9.0) -- (13.65, -9.0) -- (7.35, -9.0) -- cycle ;
                \node  (empty)    at (11, -8.5)  {\Large $I$};
    
            	\draw[hanblue, fill=hanblue, opacity=0.4] (16, -1.5) ellipse (2cm and 0.9cm); 
            	\node () at (17, -1.5) {$C^B_1$};
            	\draw[hanblue, fill=hanblue, opacity=0.4] (16, -3.5) ellipse (2cm and 0.9cm); 
            	\node () at (17, -3.5) {$C^B_2$};
                \node at (16.5, -5) {$\vdots$};
            	\draw[hanblue, fill=hanblue, opacity=0.4] (16, -7) ellipse (2cm and 0.9cm); 
            	\node () at (17, -7) {$C^B_N$};
    
            	\draw[orange, fill=orange, opacity=0.4] (16, 7) ellipse (2cm and 0.9cm); 
            	\node () at (17, 7) {$C^A_1$};
                \node at (16.5, 4.7) {$\vdots$};
            	\draw[orange, fill=orange, opacity=0.4] (16, 2) ellipse (2cm and 0.9cm); 
            	\node () at (17, 2) {$C^A_N$};
    
                \path [draw = black, rounded corners, inner sep=100pt, dotted]
                (13.85, 8.0) -- (18.1, 8.0) -- (18.1, -8.0) -- (13.85, -8.0) -- cycle ;
                \node  (empty)    at (14.8, 8.5)  {\Large $C$};
    
            	\node[circle, draw] (t1)   at (19.5, 7)  {$t_{1}$};
            	\node at (19.5, 4.7) {$\vdots$};
                \node[circle, draw] (tN)   at (19.5, 2)  {$t_{N}$};
                \node[circle, draw] (z1)   at (19.5, -1.5)  {$z_{1}$};
            	\node[circle, draw] (z2)   at (19.5, -3.5)  {$z_{2}$};
            	\node at (19.5, -5) {$\vdots$};
            	\node[circle, draw] (zN)   at (19.5, -7)  {$z_{N}$};
    
                \node[] 		(CA1)   at (16.5, 7.2) {};
                \path (CA1) edge [sloped, bend left=60] node {$0$} (t1);
                \node[] 		(CAN)   at (16.5, 2.2) {};
                \path (CAN) edge [sloped, bend left=60] node {$0$} (tN);
                \node[] 		(CB1)   at (16.5, -1.3) {};
                \path (CB1) edge [sloped, bend left=60] node {$0$} (z1);
                \node[] 		(CB2)   at (16.5, -3.3) {};
                \path (CB2) edge [sloped, bend left=60] node {$0$} (z2);
                \node[] 		(CBN)   at (16.5, -6.8) {};
                \path (CBN) edge [sloped, bend left=60] node {$0$} (zN);
    
                \node[circle, draw, accepting] (t)   at (28,0)  {$t$};
            	\node[circle, draw] (tb)   at (26,-4)  {};
            	\node[circle, draw] (ta)   at (26,4)  {};
            	\path (tb) edge [sloped] node {$1$} (t);
            	\path (ta) edge [sloped] node {$0$} (t);
    
                \node (d) at (22.5, 3) {$\ldots$};
                \node (d) at (22.5, -5) {$\ldots$};
             
                \node[circle, draw] (ta1)   at (24,6)  {};
            	\node[circle, draw] (ta2)   at (24,2)  {};
            	\path (ta1) edge [sloped] node {$0$} (ta);
            	\path (ta2) edge [sloped] node {$1$} (ta);
    
                \node[circle, draw] (tb1)   at (24,-2)  {};
            	\node[circle, draw] (tb2)   at (24,-6)  {};
            	\path (tb1) edge [sloped] node {$0$} (tb);
            	\path (tb2) edge [sloped] node {$1$} (tb);
    
                \path [draw = black, rounded corners, inner sep=100pt, dotted]
                (29.25, 1) -- (26, 7.25) -- (18.3, 8.5) -- (18.3, -8.5) -- (26, -7.25) -- (29.25, -1) -- cycle ;
                \node  (empty)    at (27.5, -6.5)  {\Large $V^{in}$};
            \end{scope}
      \end{tikzpicture}
    }
  }
\caption{Illustration of our construction of $\mathcal A$ for an arbitrary instance $A=\{a_1,\ldots, a_N\}$ and $B=\{b_1,\ldots, b_N\}$ of the OV problem. The cycles $C_1^A, \dots, C_N^A, C_1^B, \dots, C_N^B$ are expanded in Figure \ref{fig:reduction example} on a particular OV instance.}
\label{fig:reduction}
\end{figure}
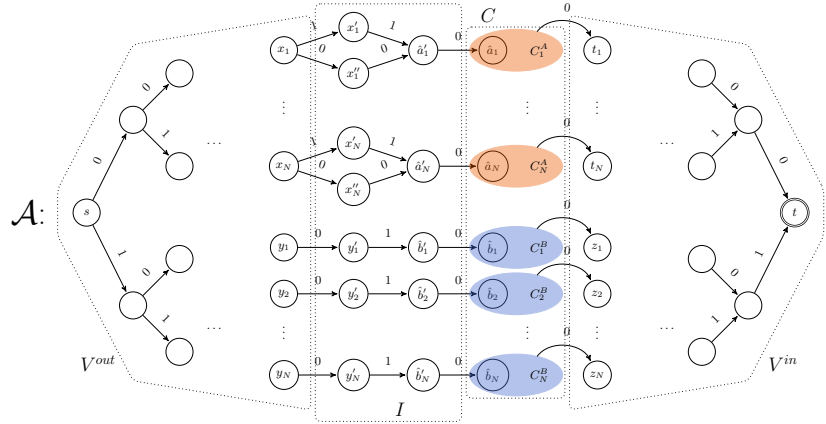

An illustration of the overall construction for an arbitrary instance of the OV problem can be found in Figure~\ref{fig:reduction}. The complete proof is deferred to the appendix. We proceed with a sketch on how we achieve the above properties. Property~(1) is achieved by connecting the above described cycles (we call the set of all cycles' nodes $C$), to the source node $s$ through a binary out-tree of logarithmic depth (we call these nodes $V^{out}$). Property~(2) is achieved by connecting $V^{out}$ to $C$ through the nodes in $I$ that ensure that nodes $u \in C^A_i,v \in C^B_j$ in the same relative positions (i.e.\ same distance from $\#$) in their cycles that correspond to orthogonal vectors $a_i$ and $b_j$ are reached by strings of alternating co-lexicographic order, i.e., for a suitable string $\tau$, $I_u$ contains two strings suffixed by $00\tau$ and $11\tau$, while $I_v$ contains a string suffixed by $01\tau$. 
This implies $\mathcal I(u) \cap \mathcal I(v) \neq \emptyset$.
Property~(3) is instead achieved by connecting one node from each cycle $C^A_i$ and $C^B_j$ (the one with out-edge $\#$) with an edge labeled 0 to a complete binary in-tree (we call these nodes $V^{in}$) with root being the only accepting state $t$. 
Observe that the
reversed automaton (i.e. the automaton obtained by reversing the direction of all the edges) is \emph{deterministic}; this ensures that any two nodes in the graph can reach $t$ through a distinct binary string, witnessing that $\mathcal A$ is indeed minimal (by the Myhill-Nerode characterization of the minimum DFA \cite{nerode1958linear}). 
Property~(4) can be easily satisfied by transforming the instance as follows. Edges labeled $0$ ($1$) are replaced with a directed path labeled with $00$ ($11$), while edges labeled $\#$ are replaced with a directed path labeled $101$. The pattern $101$ then appears only on the paths that originally corresponded to $\#$ and thus two transformed cycles match if and only if they used to match before the transformation. We note that also forward- and reverse- determinism (thus minimality) are maintained under this transformation (the nodes $\hat a'_i$ and $\hat b'_j$ in $I$ are introduced for maintaining reverse-determinism).

%
%
%
 \bibliographystyle{splncs04}
 \bibliography{ref}

\newpage

\appendix

\section{Computing $\mathcal{I}(u)$ using \cite[Sec. 4]{becker2023sorting}}\label{app:assumptions partition refinement}

We compute $\mathcal{I}(u)$ for each $u\in Q_{min}$ of the minimum DFA $\mathcal{A}_{min}$ using the partition refinement algorithm described in \cite[Sec. 4]{becker2023sorting}.

In \cite{becker2023sorting}, it is assumed that some properties hold of the input automaton. In particular, (i) every state is reachable from the source state; (ii) the source state does not have any incoming transition; and (iii) input consistency, i.e. all transitions entering in a given state bear the same label.
It is worth noting that \cite{becker2023sorting} does not assume the input automaton to be deterministic; their solution works for arbitrary NFAs.

The algorithm described in \cite[Sec. 4]{becker2023sorting} works in three steps: (1) the infima strings are computed using \cite[Alg. 2]{becker2023sorting}, (2) the suprema strings are computed using \cite[Alg. 2]{becker2023sorting}, and finally (3) the two sets of strings are merged and sorted using the suffix doubling algorithm of \cite{kop:cpm23:working} (running in linearithmic time on the pseudo-forests output by \cite[Alg. 2]{becker2023sorting}).
Since, in our paper, we run \cite[Sec. 4]{becker2023sorting} on the minimum DFA $\mathcal A_{min}$. This automaton satisfies (1), 
but it does not necessarily satisfy (ii) and (iii). We now show that any DFA $\mathcal A$ satisfying (i) and (ii) can be transformed into two NFAs $\mathcal A_{inf}$ and $\mathcal A_{sup}$ satisfying conditions (i-iii) and preserving the infima/suprema strings of $\mathcal A$, respectively. 
The transformation takes linear time in the input automaton.

We describe how to transform $\mathcal{A}$ into $\mathcal{A}_{inf}=(Q_{inf},\Sigma_{inf},\delta_{inf},s_{inf},F_{inf})$ that will be fed into \cite[Alg. 2]{becker2023sorting} to compute the infima strings of $\mathcal{A}$. 
The procedure for $\mathcal{A}_{sup}$ is symmetric.
In summary, we add two new states $s'$, $s''$, two new symbols $\$$ and $\#$ to the alphabet, and we define the transition function $\delta_{inf}$ appropriately in order to make $\mathcal{A}_{inf}$ satisfy (i-iii) while also preserving the infima strings of $\mathcal A$.

\begin{enumerate}
    \item We define $Q_{inf}=Q\cup\{s_{inf},s'\},\Sigma_{inf}=\Sigma\cup\{\$,\#\},s_{inf},F_{inf}=F$.
    \item 
    To make $\mathcal{A}_{inf}$ input-consistent, for each state we keep only the incoming transitions with the minimum label. More formally, for $u,v\in Q\cup\{s'\}$ and $a\in\Sigma$ such that $\delta(v,a)=u$, $\delta_{inf}(v,a)=u$ iff $a=\lambda(u)$, where $\lambda(u) = \min\{c\in\Sigma:\exists w\in Q\cup\{s'\}~\delta(w,c)=u\}$ is the minimum label of $u$'s incoming transitions; for convenience, we define $\delta(s',\$)=s$ where $\$$ is a special symbol with $\$\prec a$ for all $a\in\Sigma$. 
    
    \item  
    To make $\mathcal{A}_{inf}$ connected, we connect the state $s'$ to every $u\in Q_{inf}\setminus\{s_{inf},s'\}$ with a transition $\delta_{inf}(s',\lambda(u))=u$. 
    
    \item 
    Finally, we create a transition from the new source $s_{inf}$ to $s'$ with label $\#$, and the order on $\Sigma_{inf}$ is defined as: $\$\prec a\prec\#$ for all $a\in\Sigma$. This prevents the transitions defined in the previous step from affecting the infima strings of the automaton. 
\end{enumerate}

We show that the infima strings are preserved after this transformation. More precisely, for any given state $u\in Q$, let $\alpha$ and $\beta$ be the infimum string of $u$ on $\mathcal{A}_{inf}$ and $\mathcal{A}$, respectively. Then we claim that (i) $\alpha=\beta$ if $\beta$ has infinite length, and (ii) $\alpha=\#\$\beta$ if $\beta$ is a finite string. Note that the order between a finite string and a string of infinite length is also preserved because $\$$ is smaller than any symbol $a\in\Sigma$.
Let us assume for a contradiction that there exists $u\in Q$ such that this is not satisfied. We consider two cases: (i) the condition is falsified because we removed transitions in Step 2 or (ii) the infimum string is preserved in Step 2 but adding the new state $s'$ and $s_{inf}$ with new transitions in Steps 3 and 4 affects the infimum string. (i) If Step 2 affects the infimum string, $\beta$ must end with some $b(\in\Sigma)$ with $b\ne a=\min\{c\in\Sigma:\exists w\in Q~\delta(w,c)=u\}$, which implies there exists $\alpha'\in\Sigma^*\cup\Sigma^\omega$ entering $u$ such that $\alpha'$ ends with $a$ thereby being $\alpha'\prec \beta$, contradicting the definition of the infimum string. (ii) If added states and transitions in Steps 3 and 4 affect the infimum string, then we can observe that it must be $\alpha=\#\cdot\lambda(u)$, and it must also hold $\alpha\preceq\beta$ by definition of infimum string. However, this can happen only for the source state $s$ of $\mathcal{A}$. To see this, we observe that every state $u(\ne s)$ in $\mathcal{A}$ has at least one incoming transitions because it is reachable from the source, thus $\beta=\beta'\cdot\lambda(u)$ for some non-empty string $\beta'\in\Sigma^+\cup\Sigma^\omega$. Because $a\prec \#$ for every $a\in\Sigma$, we have $\beta'\prec \#$, which implies $\beta'\cdot\lambda(u)\prec\#\cdot\lambda(u)$; thus the infimum string is not affected by Step 3 except on state $s$. Note that only the state $s$ can have the empty string as its infimum string on $\mathcal{A}$, i.e. $\beta=\epsilon$, which is a finite string. Now we show $\alpha=\#\$\beta$ holds in the case that $\beta$ is finite string, which also includes the case $u=s$. For any $u$ that has a finite infimum string $\beta$, we can observe that $\delta_{inf}(s,\beta)=u$ by Lemma~\ref{lem: observation}.\ref{obs: finite inf} and there is no string $\gamma\in\Sigma^*$ such that $\gamma\prec\beta$ and $\delta_{inf}(s,\gamma)=u$ because the infimum string is preserved in Step 2. Since $\#\$$ is the unique string labeling the path from the new source $s''$ and $s$, $\#\$\beta$ is the smallest string labeling the path from $s''$ to $u$.

Transforming $\mathcal{A}$ info $\mathcal{A}_{sup}$ is symmetric. In Step 2, we keep only the transitions with maximum labels. In Step 4, we define $\#$ as the smallest symbol; i.e. $ \#\prec \$\prec a$ for every $a\in\Sigma$.

As far as the running time of the transformation is concerned, the number of states is increased by just 2 and the number of transitions is increased by at most $n+1=O(m)$, so the process takes linear time. Computing the co-lex order of the infima and suprema strings takes therefore $O(O(m)\lg (n+2))=O(m\lg n)$ time as claimed.

\section{Parameterized algorithm}\label{appendix:parameterized}

\subsection*{Proof of Lemma \ref{lem:size of pruned square A}}

We first prove the following lemma.

\begin{lemma}\label{lem: interval graph num edge}
    Let $S=\{(x_i,y_i)\ :\ i\in [n]\}$ be a (multi-)set of $n$ open intervals over domain $x_i,y_i \in [2n]$ with $x_i<y_i$. For an integer $p\ge 1$, suppose there exist no $p+1$ pairwise intersecting intervals in $S$. Let $J=\{ (i,j)\in[n]^2:i\ne j\land (x_i,y_i)\cap(x_j,y_j) \neq\emptyset\}$ be the set of (indices of) intersecting intervals of $S$. Then $|J|\le 2n(p-1)$. 
\end{lemma}
\begin{proof}
    Consider the open-interval graph $G$ constructed from $S$ (that is, the undirected graph with set of nodes $S$ and edges connecting intersecting intervals).
    Since no $p+1$ intervals in $S$ are pairwise intersecting, there is no clique of size $p+1$ in such an interval graph. 
    
    We apply the transformation of \cite[Lem. 12]{kop:cpm23:working} to the set of open intervals $S = \{(x_i,y_i)\ :\ i\in[n]\}$ 
    mapping each $(x_i,y_i)$ into a closed interval $[x_i,y_i]$ over domain $x'_i,y'_i \in [4n+1]$ with the property that $(x_i,y_i) \cap (x_j,y_j) \neq \emptyset$ if and only if $[x'_i,y'_i]\cap [x'_j,y'_j] \neq \emptyset$. 
    Let $S' = \{[x'_i,y'_i]\ :\ i\in[n]\}$ be the transformed interval set.
    As a result, no $p+1$ intervals in $S'$ are pairwise intersecting.

    The number of edges of a closed-interval graph (i.e. constructed from closed intervals) that does not contain a clique of size $p'=p+1$ is known to be not greater than $\frac{1}{2}(p'-1)(p'-2)+(n-p'+1)(p'-2) \le n(p-1)$ \cite[Sec.~2]{abbott79:turan:intervalgraph}.
    Then, by \cite[Lem. 12]{kop:cpm23:working}, also the open-interval graph $G$ corresponding to intervals $S$ contains at most $n(p-1)$ edges. 
    Note that the size of $J$ is equal to twice the number of edges of $G$. Our claim follows. \qed
\end{proof}

We proceed with the proof of Lemma \ref{lem:size of pruned square A}.
Recall that $p=width(\mathcal{A}_{min})$ implies that there exist no $p+1$ pairwise incomparable states with respect to the order $<_{\mathcal{A}_{min}}$ defined in Definition~\ref{def: <_A}.
From Lemma~\ref{lem_W order2}, this is equivalent to the fact that there are no $p+1$ pairwise intersecting open intervals $\mathcal{I}(u) = (\inf I_u, \sup I_u)$ in $\mathcal{A}_{min}$. 
From Lemma~\ref{lem: interval graph num edge}, 
the set $\{\mathcal{I}(u)=(\inf I_u, \sup I_u)\ :\ u\in Q_{min}\}$ contains at most $2n(p-1)$ pairs of intersecting intervals, implying that the number of states of $\mathcal {\hat A}_{min}^2$ is $|\{(u,v)\in Q_{min}^2: u\ne v\land\mathcal{I}(u)\cap\mathcal{I}(v)\ne\emptyset\}| \leq 2n(p-1)$.

Now we shall show the number of transitions of $\hat{\mathcal{A}}^2_{min}$ is at most $2m(p-1)$.
For each $a\in\Sigma$, consider the set $Q_a=\{u\in Q_{min}: \exists v\in Q_{min}:\delta'(u,a)=v\}$ of states of $Q_{min}$ that have an outgoing transition labeled with $a$.
Since $Q_a$ is a subset of $Q_{min}$, there are no $p+1$ pairwise intersecting intervals in the set $\{\mathcal{I}(u)=(\inf I_u, \sup I_u)\ :\ u\in Q_a\}$. By Lemma~\ref{lem: interval graph num edge}, this implies that $Q_a$ contains at most $2|Q_a|(p-1)$ pairs of intersecting intervals.
Observe that the pruned square automaton $\hat{\mathcal{A}}^2_{min}$ has a transition from $(u,v)\in Q^2_{min}$ to $(u',v')\in Q^2_{min}$ labeled by $a$ only if $(u,v)\in Q_a^2$ and $\mathcal{I}(u)\cap\mathcal{I}(v)\ne\emptyset$.
Since $\hat{\mathcal{A}}^2_{min}$ is a DFA, the state $(u,v)\in Q^2_{min}$ and letter $a$ uniquely determine this transition, so the number of transitions with label $a$ in $\hat{\mathcal{A}}^2_{min}$ is at most $2|Q_a|(p-1)$.
The total number of transitions in $\hat{\mathcal{A}}^2_{min}$ is therefore at most $\sum_{a\in\Sigma} 2|Q_a|(p-1)\le 2m(p-1)$ because $\sum_{a\in\Sigma}|Q_a|= m$.
\qed

\paragraph{\textbf{The algorithm}}

We now present an algorithm to build $\mathcal {\hat A}_{min}^2$ in a time proportional to its size $O(mp)$.

We first 
show how to build the $O(np)$ states of $\mathcal {\hat A}_{min}^2$ in $O(np)$ time.
We sort intervals\footnote{Recall that we represent each $\inf I_u$ and $\sup I_u$ with an integer in the range $[1,2n]$: the co-lexicographic rank of those strings in the set $\{\inf I_u:u\in Q\}\cup\{\sup I_u:u\in Q\}$. As a result, we can radix-sort these intervals in $O(n)$ time.}
$\mathcal I(u) = (\inf I_u,\sup I_u)$ for $u\in Q$ in increasing order by their first components $\inf I_u$. Let $\mathcal I(u_1), \dots, \mathcal I(u_n)$ be the resulting order. Notice that, given a state $u_i$, all states $u_j$ such that $\mathcal I(u_i) \cap \mathcal I(u_j) \neq \emptyset$ are adjacent in this order.
Then, the $O(np)$ states $(u,v)$ of $\mathcal {\hat A}_{min}^2$ can be built in $O(np)$ time with the following algorithm: (1) initialize $i=1$ and $j=2$. (2) If $j\leq n$ and $\mathcal I(u_i) \cap \mathcal I(u_j) \neq \emptyset$, then create states\footnote{Note that $\mathcal {\hat A}_{min}^2$ is symmetric: if $(u,v)$ is a state, then also $(v,u)$ is a state (a similar property holds for its transition function). Although we could further prune it by removing this symmetry, this would slightly complicate our notation.} $(u_i,u_j)$ and $(u_j,u_i)$ and increment $j$. 
Otherwise, increment $i$ and set $j=i+1$. (3) If $i<n$, go to step (2). 

Each step of the algorithm either creates a new state of $\mathcal {\hat A}_{min}^2$ or increments $i$ ($i$ can be incremented at most $n$ times). It follows that its running time is proportional to the number of states, $O(np)$.

To conclude, we show how to build the $O(mp)$ transitions of $\mathcal {\hat A}_{min}^2$ in $O(mp)$ time; let us denote with $\delta''$ the transition function of  $\mathcal {\hat A}_{min}^2$. 
Let $\mathcal I(u_1), \dots, \mathcal I(u_n)$ be the order defined above. For each $a\in \Sigma$, we build an array $L_a$ defined as follows: $L_a$ contains all states $u_i$ such that $u_i$ has an outgoing edge labeled with $a$, sorted by increasing index $i$. 
Let $m_a$ denote the number of nodes in array $L_a$.
Clearly, the total number of nodes in all arrays is $\sum_{a\in\Sigma} m_a  = m$ since we add a node in one of those arrays for each edge of $\mathcal A_{min}$; moreover, those arrays can be built in $O(m)$ time by simply visiting all nodes in the order $u_1, \dots, u_n$: when visiting $u_i$, for all $a\in \Sigma$ labeling an outgoing edge of $u_i$, append $u_i$ at the end of $L_a$.

To build the transition function $\delta''$ of $\mathcal {\hat A}_{min}^2$ in $O(mp)$ time, for each $a\in \Sigma$ we run the following algorithm. (1) initialize $i=1$ and $j=2$. (2) We distinguish three cases (A-C):

(A) If $j\leq m_a$ and $(L_a[i],L_a[j])$ and $(\delta_{min}(L_a[i],a),\delta_{min}(L_a[j],a))$ are both states of $\mathcal {\hat A}_{min}^2$ (as computed above), create the two symmetric transitions 
$$
\begin{array}{l}
\delta''((L_a[i],L_a[j]),a) = (\delta_{min}(L_a[i],a),\delta_{min}(L_a[j],a))\\
\delta''((L_a[j],L_a[i]),a) = (\delta_{min}(L_a[j],a),\delta_{min}(L_a[i],a))
\end{array}
$$  
and increment $j$.

(B) If $j\leq m_a$ and $(L_a[i],L_a[j])$ is a state 
while $(\delta_{min}(L_a[i],a),\delta_{min}(L_a[j],a))$ is not, just increment $j$. 

(C) If none of $(L_a[i],L_a[j])$ and $(\delta_{min}(L_a[i],a),\delta_{min}(L_a[j],a))$ are states in $\mathcal{A}^2_{min}$ or if $j > m_a$, increment $i$ and set $j=i+1$. 

Each step of the algorithm either creates a new transition of $\mathcal {\hat A}_{min}^2$ (case A), visits a state $(L_a[i],L_a[j])$ of $\mathcal {\hat A}_{min}^2$ (case B: note that in this case $j$ is incremented, so no state is visited twice), or increments $i$ (case C: note that $i$ can be incremented at most $m_a$ times). It follows that its running time is proportional to the number of states plus the number of transitions, $O(mp)$.\qed

\section{Reduction Details}\label{appendix: reduction}

We start with a detailed description of the DFA $\mathcal A= (Q, \Sigma, \delta, s, F)$ constructed for an OV instance  $A=\{a_1,\ldots, a_N\}$ and $B=\{b_1,\ldots, b_N\}$. 
The states $Q$ of $\mathcal A$ consist of four disjoint sets $V^{out}$, $I$, $C$, and $V^{in}$. We will now define these four sets and the transitions connecting them (thus defining $\delta$). We refer the reader back to Figure~\ref{fig:reduction} for an illustration.
Throughout the description, for an integer $i\in[N]$, we denote by $\rho(i)$ the bit-string of length $\ell$ that represents the integer $i - 1$ in binary.

\begin{description}
    \item[$V^{out}$:] The set $V^{out}$ is connected as a complete binary out-tree of depth $\ell + 1$ with root being the source node $s$. This tree has $2N$ leaves that we call $x_1,\ldots, x_N$ and $y_1,\ldots, y_N$ and is such that a leaf node $x_i$ is reached from $s$ by a unique path labeled $0\rho(i)$, while a leaf node $y_j$ is reached from $s$ by a unique path labeled $1\rho(j)$.
    \item[$I$:] The set $I$ consists of $5N$ nodes, $x_i'$, $x_i''$, and $\hat a_i'$ for $i\in [N]$ and $y_j'$ and $\hat b_j'$ for $j\in [N]$. A node $\hat a_i'$ is reachable from $x_i\in V^{out}$ by two paths through $x_i'$ and $x_i''$ labeled $11$ and $00$, respectively. A node $\hat b_j'$ is reachable from $y_j\in V^{out}$ by a path labeled $01$ through $y_j'$.
     \item[$C$:] The set $C$ is composed of $2N$ disjoint sets (cycles), $C_i^A$ for $i\in[N]$ and $C_j^B$ for $j\in[N]$. Each set $C_i^A$ contains a node $\hat a_i$ and each set $C_j^B$ contains a node $\hat b_j$. Every node $\hat a_i$ ($\hat b_j$) is reachable from $\hat a'_i\in I$ ($\hat b'_j\in I$) by an edge labeled 0. Each set $C_i^A$ and $C_j^B$ consists of $d + \ell + 1$ nodes. These nodes are connected in a cycle-like way as illustrated in Figure~\ref{fig:reduction det}: (1)~From node $\hat a_i$, we can exactly read the bit string $a_i$ that is equal to the $d$ bits in the input vector $a_i$ (note that we slightly abuse notation here by denoting with $a_i$ both the string and the vector) by passing through the subsequent states $(\hat a_i = a_i^0), a_i^1, \ldots, a_i^d$. From $a_i^d$, we can then read any sequence of $d$ bits in $\{0,1\}^\ell$ by passing through the states $q_i^1,\ldots, q_i^\ell$ followed by the character $\#$ and arriving back at node $\hat a_i$. The latter sequence of nodes is the gadget $Q_i$ from Figure~\ref{fig:reduction det}. (2)~From node $\hat b_j$, we can also read strings of length $d$ that however depend on the input vector $b_j$, this time by passing through subsequent states $(\hat b_j=b_j^0), b_j^1, \ldots, b_j^d$. If $b_j[r]=1$, node $b_{j-1}^r$ is connected to node $b_j^r$ with one edge labeled with 0. If $b_j[r]=0$, node $b_{j-1}^r$ is connected to node $b_j^r$ with two edges, labeled with 0 and 1 respectively. From $b_j^d$, we can then read the string $\rho(j)$ by passing through the states $p_j^1,\ldots, p_j^\ell$ followed by the character $\#$ and arriving back at node $\hat b_j$. The latter sequence of nodes is the gadget $P_j$ from Figure~\ref{fig:reduction det}.
    \item[$V^{in}$:] The set $V^{in}$ is connected as a complete binary in-tree of depth $\ell + 1$ with root being a state $t$. This tree again has $2N$ leaves that we call $t_1,\ldots, t_N$ and $z_1,\ldots, z_N$ and is such that a leaf node $t_i$ can reach the root $t$ by a unique path labeled $\rho(i)0$, while a leaf node $z_j$ can reach $t$ by a unique path labeled $\rho(j)1$. Every node $t_i$ is reachable from $q_i^\ell\in C^A_i$ by an edge labeled $0$ and every node $z_j$ is reachable from $p_j^\ell\in C^B_j$ by an edge labeled $0$.
\end{description} 
Finally we define the set of final states as $F:=\{t\}$. We proceed with the following observation:
The constructed DFA $\mathcal A$ is in fact minimal, i.e., $\mathcal A = \mathcal A_{min}$. This is an easy observation from inspecting the reversed automaton $\mathcal A^r$ that has the same set of states, $t$ as source node, and all transitions reversed compared to $\mathcal A$. We observe that $\mathcal A^r$ is also deterministic and thus every state in $\mathcal A^r$ is reached by a unique string from the source state $t$. It follows that every two states $u,v$ are clearly distinguishable in $\mathcal A$ by the Myhill-Nerode relation\footnote{Two states are Myhill-Nerode equivalent --- i.e.\ they can be collapsed in the minimum DFA --- if and only if they allow reaching final states with the same set of strings.} \cite{nerode1958linear} as they can reach the only final state $t$ using a unique string.

\begin{figure}[ht]
  \centering{
    \resizebox{0.8\columnwidth}{!}{
      \begin{tikzpicture}[
          scale=.7,
          ->,
          >=stealth',
          shorten >=1pt,
          auto,
          semithick,
          every node/.style={minimum size=8mm}
        ]

        \begin{scope}[yshift=6cm, xshift=4cm]
        	\node[orange] at (-0.5, 3.85) {\LARGE $C^A_i:$};
        	\node[circle, draw] (A)   at (0,0)  {$\hat a_i$};
        	\node[circle, draw] (B)   at (4 - 2.8284, -2.8284) {$a_i^1$};
        	\node[circle, draw] (C)   at (4, -4)  {$a_i^2$};
        	\node[circle, draw] (D)   at (4 + 2.8284, -2.8284) {$a_i^3$};
        	\node[circle, draw] (E)   at (8, 0)  {$a_i^4$};
        	\node[circle, draw] (F)   at (4 + 2.8284, 2.8284) {};
        	\node[circle, draw] (G)   at (4, 4)  {};
        	\node[circle, draw] (H)   at (4 - 2.8284, 2.8284)  {$a_i^d$};
        	\path (A) edge [sloped] node {$a_i[1]$} (B)
	              (B) edge [sloped] node {$a_i[2]$} (C)
	              (C) edge [sloped] node {$a_i[3]$} (D)
	              (D) edge [sloped] node {$a_i[4]$} (E)
	              (E) edge[white] [sloped] node[black] {$\ldots$} (F)
	              (F) edge [sloped] node {$a_i[d-1]$} (G)
	              (G) edge [sloped] node {$a_i[d]$} (H)
	              (H) edge[line width=1mm, dashed] [sloped, color=darkgreen] node {$Q_i$} (A);
        \end{scope}

        \begin{scope}[yshift=6cm, xshift=15cm]
        	\node[hanblue] at (-0.5, 3.85) {\LARGE $C^B_j:$};
        	\node[circle, draw] (A)   at (0,0)  {$\hat b_j$};
        	\node[circle, draw] (B)   at (4 - 2.8284, -2.8284) {$b_j^1$};
        	\node[circle, draw] (C)   at (4, -4)  {$b_j^2$};
        	\node[circle, draw] (D)   at (4 + 2.8284, -2.8284) {$b_j^3$};
        	\node[circle, draw] (E)   at (8, 0)  {$b_j^4$};
        	\node[circle, draw] (F)   at (4 + 2.8284, 2.8284) {};
        	\node[circle, draw] (G)   at (4, 4)  {};
        	\node[circle, draw] (H)   at (4 - 2.8284, 2.8284)  {$b_j^d$};
        	\path (A) edge [sloped] node {$0$} (B)
	              (B) edge[bend left] [sloped] node {$0$} (C)
	              (C) edge [sloped] node {$0$} (D)
	              (D) edge[bend left] [sloped] node {$0$} (E)
	              (E) edge[white] [sloped] node[black] {$\ldots$} (F)
	              (F) edge [bend right] [sloped]node {$0$} (G)
	              (G) edge [sloped] node {$0$} (H)
	              (H) edge[line width=1mm, dashed] [sloped, color=darkyellow] node {$P_j$} (A);

	        \path (B) edge[bend right] [sloped] node {$1$} (C)
	        	  (D) edge[bend right] [sloped] node {$1$} (E)
	        	  (F) edge [bend left] [sloped] node {$1$} (G);

        \end{scope}

        \begin{scope}[yshift=-0.5cm, xshift=4cm]
        	\node[darkgreen] at (-2, 0) {\Large $Q_i:$};
	        \node[circle, draw] (A)   at (0,0)  {$a_i^d$};
	        \node[circle, draw] (B)   at (4,0)  {$q_i^1$};
	        \node[circle, draw] (C)   at (8,0) {$q_i^2$};
	        \node[circle, draw] (D)   at (12,0) {};
	        \node[] 		(D)   at (10.5,0)  {$\ldots$};
	        \node[circle] 	(E)   at (12,0)  {};
	        \node[circle, draw] 	(F)   at (16,0)  {$q_i^{\ell}$};
            \node[circle, draw] 	(G)   at (20,0)  {$\hat a_i$};

	        \path 
	              (A) edge [above, bend left] node {$0$} (B)
	              (B) edge [above, bend left] node {$0$} (C)
	              (C) edge [above, bend left] node {} (D)
                  (E) edge [above, bend left] node {$0$} (F);

	        \path (A) edge [above, bend right] node {$1$} (B)
	              (B) edge [above, bend right] node {$1$} (C)
	              (C) edge [above, bend right] node {} (D)
                  (E) edge [above, bend right] node {$1$} (F);
                \path (F) edge [above] node {$\#$} (G);
        \end{scope}

        \begin{scope}[yshift=-3cm, xshift=4cm]
	        \node[darkyellow] at (-2, 0) {\Large $P_j:$};
	        \node[circle, draw] (A)   at (0,0)  {$b_j^d$};
	        \node[circle, draw] (B)   at (4,0)  {$p_j^1$};
	        \node[circle, draw] (C)   at (8,0) {$p_j^2$};
	        \node[circle, draw] (D)   at (12,0) {};
	        \node[] 		(D)   at (10.5,0)  {$\ldots$};
	        \node[circle] 	(E)   at (12,0)  {};
	        \node[circle, draw] 	(F)   at (16,0)  {$p_j^\ell$};
            \node[circle, draw] 	(G)   at (20,0)  {$\hat b_j$};

	              \path (A) edge [above] node {$\rho(j)[1]$} (B)
	              (B) edge [above] node {$\rho(j)[2]$} (C)
	              (C) edge [above] node {} (D)
                    (E) edge [above] node {$\rho(j)[\ell]$} (F)
                    (F) edge [above] node {$\#$} (G);
	    \end{scope}

      \end{tikzpicture}
    }
  }
\caption{Illustration of $C_i^A$ and $C_j^B$ containing the subgraphs $Q_i$ and $P_j$, respectively.}
\label{fig:reduction det}
\end{figure}
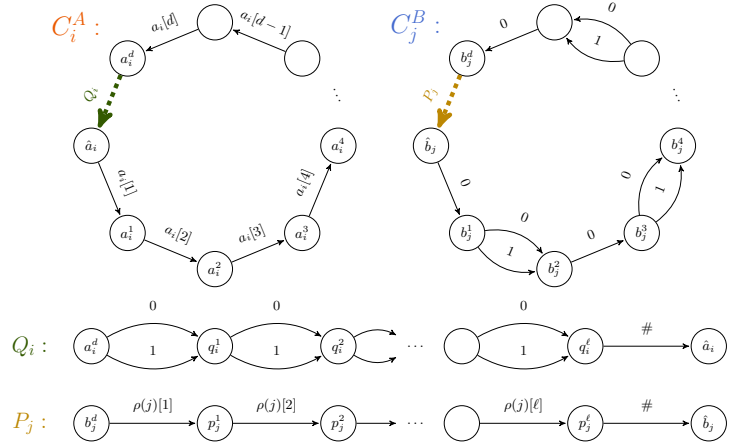

\subsection*{Proof of Proposition~\ref{prop: reduction}}
We are now ready to prove Proposition~\ref{prop: reduction}. We start with the following lemma.
\begin{lemma}\label{lem: cycles AB}
    Let $(u_1,v_1) \rightarrow (u_2,v_2) \rightarrow \dots \rightarrow (u_k,v_k) \rightarrow (u_1,v_1)$ be a cycle in $\mathcal{A}^2$ such that, for $1\le \forall i \le k$, $u_i\neq v_i$. Then, there exist $r, s \in [N]$ and $\ell\in [k]$ such that $\hat a_r=u_\ell$ and $\hat b_s = v_\ell$ (or vice versa) and furthermore $\{u_1, \ldots, u_k\} = C_r^A$ and $\{v_1, \ldots, v_k\} = C_s^B$ (or vice versa).
\end{lemma}
\begin{proof}
    For simplicity let us call $C=\{u_1, \ldots, u_k\}$ and $C'=\{v_1, \ldots, v_k\}$ and let us take the indices of those nodes modulo $k$, i.e., $u_{i+k} = u_i$.
    We first notice that the only directed cycles in $\mathcal A$ are $C_r^A$ for $r\in[N]$ and $C_s^B$ for $s\in[N]$ and they are all of length $d+\ell$. Hence $k=d+\ell+1$. As the character $\#$ appears only once in every such cycle, it follows that there exists $\ell\in [k]$ such that $u_\ell, v_\ell\in \{\hat a_r:r\in [N]\} \cup \{\hat b_s:s\in [N]\}$.
    We now proceed by showing that it cannot be that $C = C_r^A$ and $C' = C_{r'}^A$ for some $r,r'\in[N]$ or $C = C_s^B$ and $C' = C_{s'}^B$ for some $s,s'\in[N]$. First assume that $C = C_r^A$ and $C' = C_{r'}^A$ for some $r,r'\in[N]$. Notice that $r\neq r'$ by the assumption that $u_i\neq v_i$ for $1\le \forall i \le k$. We now directly get a contradiction as the labels of the edges $(u_\ell, u_{\ell +1}), \ldots, (u_{\ell + d - 1}, u_{\ell + d})$ and $(v_\ell, v_{\ell +1}), \ldots, (v_{\ell + d - 1}, v_{\ell + d})$ cannot match by the assumption that $A$ contains no two identical vectors. Now assume that $C = C_s^B$ and $C' = C_{s'}^B$ for some $s,s'\in[N]$. This analogously yields to a contradiction as $s\neq s'$ and the paths $P_s$ and $P_{s'}$ cannot have matching labels. It thus follows that $C = C_r^A$ for some $r\in[N]$ and $C' = C_{s}^B$ for some $s\in[N]$ (or vice versa). Consequently also $\hat a_r=u_\ell$ and $\hat b_s = v_\ell$ (or vice versa) and this completes the proof.
\end{proof}

\subsubsection*{$(\Rightarrow)$}
As the OV instance is a YES-instance, there exist vectors $a_r\in A$ and $b_s\in B$, $r,s\in[N]$ such that $a_r[i]\cdot b_s[i]=0$ for all $i\in [d]$. Let $C_r^A$ and $C_s^B$ be the node sets corresponding to $a_r$ and $b_s$ in $\mathcal A$ and let  $u_1, \ldots, u_{\ell + d + 1}$ and $v_1, \ldots, v_{\ell + d + 1}$ be the nodes in $C_r^A$ and $C_s^B$, respectively, starting from $u_1=\hat a_r$ and $v_1= \hat b_s$. We now need to show that these two cycles can be traversed using the same labels. For $i\in [d]$, transition  $(u_i, u_{i+1})$ is labeled $a_r[i]$. If $b_s[i]=0$, there exist transitions from $v_i$ to $v_{i+1}$ with both labels $0$ and $1$ and so we can always choose a transition matching $a_r[i]$. If $b_s[i]=1$, there exist only a transition labeled $0$ from $v_i$ to $v_{i+1}$, however in this case it must hold that $a_r[i]=0$ and thus the transition from $u_i$ to $u_{i+1}$ is again matched. Now for $d+1 \le i \le d + \ell$, node $u_{i+1}$ can be reached from node $u_i$ using any suitable transition. Then, it is clear that both cycles on the last transition $(u_{d+\ell+1}, u_1)$ and $(v_{d+\ell+1}, v_1)$ do the transition $\#$ resepctively. We call the string that can be read in both cycles $\gamma$. It remains to show that, for $1\le \forall i \le \ell + d + 1$, (i) $u_i\neq v_i$ and (ii) $\mathcal{I}(u_i)\cap\mathcal{I}(v_i)\ne\emptyset$. Property (i) is clearly true. For (ii), we let $s_1 = 0\rho(r)000\gamma[1..i]$, $s_2 = 1\rho(s)010\gamma[1..i]$, and $s_3 = 0\rho(r)110\gamma[1..i]$ and observe that $s_1 < s_2 < s_3$ and $s_1, s_3\in I_{u_i}$ as well as $s_2\in I_{v_i}$. It thus follows that $\mathcal{I}(u_i)\cap\mathcal{I}(v_i)\ne\emptyset$.

\subsubsection*{$(\Leftarrow)$}
Let $(u_1,v_1) \rightarrow (u_2,v_2) \rightarrow \dots \rightarrow (u_k,v_k) \rightarrow (u_1,v_1)$ be a cycle in $\mathcal{A}^2$ such that, for all $i \in [k]$, (i) $u_i\neq v_i$ and (ii) $\mathcal{I}(u_i)\cap\mathcal{I}(v_i)\ne\emptyset$. From Lemma~\ref{lem: cycles AB} it follows that there exist $r, s \in [N]$ and $\ell\in [k]$ such that $a_r=u_\ell$ and $b_s = v_\ell$ (or vice versa) and furthermore $\{u_1, \ldots, u_k\} = C_r^A$ and $\{v_1, \ldots, v_k\} = C_s^B$ (or vice versa). Furthermore $k=d+\ell+1$. W.l.o.g., assume that $a_r=u_1$, $b_s = v_1$ and $\{u_1, \ldots, u_k\} = C_r^A$ and $\{v_1, \ldots, v_k\} = C_s^B$. From $(u_1,v_1) \rightarrow (u_2,v_2) \rightarrow \dots (u_{d+1}, v_{d+1})$ we conclude that the label of $(v_i, v_{i+1})$ is equal to $a_r[i]$ for every $i\in [d]$. We are now ready to argue that $a_r[i] \cdot b_s[i]=0$ for every $i\in [d]$ and thus $a_r$ and $b_s$ are orthogonal. If $a_r[i]=0$, there is nothing to show. If $a_r[i]=1$, the previous claim yields that the label of $(v_i, v_{i+1})$ is $1$, which by definition of the transitions in $C_s^B$ implies that $b_s[i]=0$. Hence $a_r[i] \cdot b_s[i]=0$ also in this case.

\subsubsection*{Alphabet size 2}
Finally, we observe that the constructed DFA $\mathcal A$ is over an alphabet of size $3$ as we introduced the letter $\#$ in addition to $0$ and $1$. We now apply the following transformation to $\mathcal A$ in order to obtain a DFA over the binary alphabet $\{0,1\}$. We replace every edge labeled with $0$ with a directed path of length 2 labeled with $00$, every edge labeled with $1$ with a directed path of length 2 labeled with $11$, and the edges labeled with $\#$ with a directed path of length 3 labeled with $101$. Since the pattern $101$ appears only on the paths that originally corresponded to the character $\#$, it is easy to see that two transformed cycles match if and only if they used to match before the transformation.
We note that this transformation preserves determinism of $\mathcal A$ as the nodes with outgoing edges labeled $\#$ have only two outgoing transitions, the one labeled $\#$ and one labeled $0$ going to the nodes in $V^{in}$. Recall that we used also determinism of the reversed graph $\mathcal A^r$ in order to argue that $\mathcal A$ is minimal. It is easy to see that the transformation also maintains determinism of the reversed graph as the nodes with ingoing transitions labeled $\#$ have only one second ingoing transition that is labeled $0$ (the one from the nodes in $I$). 

\end{document}